\pdfoutput=1
\documentclass[entropy,article,submit,moreauthors,pdftex,10pt,a4paper]{mdpi-arxiv} 
\firstpage{1} 
\articlenumber{x}
\doinum{10.3390/------}
\pubvolume{xx}
\pubyear{2018}
\copyrightyear{2018}
\history{Received: 24 March 2018; Accepted: 11 May 2018}

\usepackage[normalweight=Book]{fdsymbol}
\let \Sqcup\relax
\usepackage{fixltx2e}
\usepackage{blindtext}
\usepackage{verbatim}
\usepackage{url}
\usepackage{color}
\usepackage{nicefrac}
\usepackage{graphicx}
\usepackage{multirow}
\usepackage{rotate}
\usepackage{stackengine} 
\usepackage{bm} 
\usepackage{thmtools}
\usepackage{thm-restate}
\usepackage{relsize} 
\usepackage[ruled,noend]{algorithm2e} 
\usepackage{multicol}
\usepackage{wrapfig}
\usepackage{ulem} 
\usepackage{float}
\usepackage{subcaption}

\allowdisplaybreaks[1]

\Title{A Game-Theoretic Approach to Information-Flow Control via Protocol Composition}

\Author{M\'{a}rio S. Alvim $^{1}$\orcidA{}, Konstantinos Chatzikokolakis $^{2,3}$\orcidB{}, Yusuke Kawamoto$^{4}$\orcidC{} and Catuscia Palamidessi $^{5,2}$\orcidD{}}

\AuthorNames{M\'{a}rio S. Alvim, Konstantinos Chatzikokolakis, Yusuke Kawamoto and Catuscia Palamidessi.}

\address{%
$^{1}$ \quad Universidade Federal de Minas Gerais, Brazil; \\
$^{2}$ \quad \'{E}cole Polytechnique, France; \\ 
$^{3}$ \quad CNRS, France; \\ 
$^{4}$ \quad AIST, Japan; \\ 
$^{5}$ \quad INRIA, France;}

\abstract{
In the inference attacks studied in Quantitative Information Flow (QIF), the attacker typically tries to interfere with the system in the attempt to increase its leakage of secret information. The defender, on the other hand, typically tries to decrease leakage by introducing some controlled noise. This noise introduction can be modeled as a type of protocol composition, i.e., a probabilistic choice among different protocols, and its effect on the amount of leakage depends heavily on whether or not this choice is visible to the attacker. In this work, we consider operators for modeling visible and hidden choice in protocol composition, and we study their algebraic properties. We then formalize the interplay between defender and attacker in a game-theoretic framework adapted to the specific issues of QIF, where the payoff is information leakage. We consider various kinds of leakage games, depending on whether players act simultaneously or sequentially, and on whether or not the choices of the defender are visible to the attacker. In the case of sequential games, the choice of the second player is generally a function of the choice of the first player, and his/her probabilistic choice can be either over the possible functions (mixed strategy) or it can be on the result of the function (behavioral strategy). We show that when the attacker moves first in a sequential game with a hidden choice, then behavioral strategies are more advantageous for the defender than mixed strategies. This contrasts with the standard game theory, where the two types of strategies are equivalent. Finally, we establish a hierarchy of these games in terms of their information leakage and provide methods for finding optimal strategies (at the points of equilibrium) for both attacker and defender in the various cases.
}


\begin{document}
\newcommand{\review}[1]{#1}
\newif\ifcommentson\commentsonfalse
\def\mywidth{.9}
\def\mywidthRep{.8}
\ifcommentson
\newcommand{\commentCP}[1]{\begin{center} \parbox{\mywidth\textwidth}{\textbf{\textcolor{black}{Comment C.}} \textcolor{red}{#1 }}\end{center}}
\newcommand{\commentKC}[1]{\begin{center} \parbox{\mywidth\textwidth}{\textbf{\textcolor{black}{Comment K.}} \textcolor{red}{#1} }\end{center}}
\newcommand{\commentYK}[1]{\begin{center} \parbox{\mywidth\textwidth}{\textbf{\textcolor{black}{Comment Y.}} \textcolor{red}{#1} }\end{center}}
\newcommand{\commentMA}[1]{\begin{center} \parbox{\mywidth\textwidth}{\textbf{\textcolor{black}{Comment M.}} \textcolor{red}{#1} }\end{center}}
\newcommand{\replyCP}[1]{\begin{center} \parbox{\mywidthRep\textwidth}{\textbf{Reply C.} \textcolor{blue}{#1} }\end{center}}
\newcommand{\replyKC}[1]{\begin{center} \parbox{\mywidthRep\textwidth}{\textbf{Reply K.} \textcolor{blue}{#1} }\end{center}}
\newcommand{\replyYK}[1]{\begin{center} \parbox{\mywidthRep\textwidth}{\textbf{Reply Y.} \textcolor{blue}{#1} }\end{center}}
\newcommand{\replyMA}[1]{\begin{center} \parbox{\mywidthRep\textwidth}{\textbf{Reply M.} \textcolor{blue}{#1} }\end{center}}
\newcommand{\commentC}[1]{\marginpar{\footnotesize \color{red} {\bf C:} \textsf{\scriptsize #1}}}
\newcommand{\commentK}[1]{\marginpar{\footnotesize \color{red} {\bf K:} \textsf{\scriptsize #1}}}
\newcommand{\commentY}[1]{\marginpar{\footnotesize \color{red} {\bf Y:} \textsf{\scriptsize #1}}}
\newcommand{\commentM}[1]{\marginpar{\footnotesize \color{red} {\bf M:} \textsf{\scriptsize #1}}}
\newcommand{\replyC}[1]{\marginpar{\footnotesize \color{red} {\bf C:} \textsf{\scriptsize #1}}}
\newcommand{\replyK}[1]{\marginpar{\footnotesize \color{red} {\bf K:} \textsf{\scriptsize #1}}}
\newcommand{\replyY}[1]{\marginpar{\footnotesize \color{red} {\bf Y:} \textsf{\scriptsize #1}}}
\newcommand{\replyM}[1]{\marginpar{\footnotesize \color{red} {\bf M:} \textsf{\scriptsize #1}}}
\newcommand{\colorB}[1]{\textcolor{blue}{#1}}
\else
\newcommand{\commentCP}[1]{}
\newcommand{\commentKC}[1]{}
\newcommand{\commentYK}[1]{}
\newcommand{\commentMA}[1]{}
\newcommand{\replyCP}[1]{}
\newcommand{\replyKC}[1]{}
\newcommand{\replyYK}[1]{}
\newcommand{\replyMA}[1]{}
\newcommand{\commentC}[1]{}
\newcommand{\commentK}[1]{}
\newcommand{\commentY}[1]{}
\newcommand{\commentM}[1]{}
\newcommand{\replyC}[1]{}
\newcommand{\replyK}[1]{}
\newcommand{\replyY}[1]{}
\newcommand{\replyM}[1]{}
\newcommand{\colorB}[1]{#1}
\fi

\newcommand\Cite[1]{[\textbf{\color{red}{#1}}]} 


\newcommand{\cals}{\mathcal{S}}
\newcommand{\calx}{\mathcal{X}}
\newcommand{\caly}{\mathcal{Y}}
\newcommand{\calc}{\mathcal{C}}
\newcommand{\calg}{\mathcal{G}}
\newcommand{\cala}{\mathcal{A}}
\newcommand{\cald}{\mathcal{D}}
\newcommand{\calw}{\mathcal{W}}
\newcommand{\calad}{\mathcal{A}\,{\rightarrow}\,\mathcal{D}}
\newcommand{\calda}{\mathcal{D}\,{\rightarrow}\,\mathcal{A}}
\newcommand{\cali}{\mathcal{I}}
\newcommand{\calj}{\mathcal{J}}
\newcommand{\reals}{\mathbb{R}}
\newcommand{\distr}{\mathbb{D}}
\newcommand{\dxy}{\distr\mathcal{\calx\times\caly}}
\newcommand{\infoset}{\textit{K}}
\newcommand{\infosetd}{\infoset_{\mathsf{d}}}
\newcommand{\infoseta}{\infoset_{\mathsf{a}}}
\newcommand{\calsa}{\mathcal{S}_{\mathsf{a}}}
\newcommand{\calsd}{\mathcal{S}_{\mathsf{d}}}

\newcommand{\psd}{\textit{s}_{\sf d}}
\newcommand{\psa}{\textit{s}_{\sf a}}
\newcommand{\msd}{{\sigma_{\sf d}}}
\newcommand{\msa}{{\sigma_{\sf a}}}
\newcommand{\fsd}{{\phi_{\sf d}}}
\newcommand{\fsa}{{\phi_{\sf a}}}
\newcommand{\payd}{\textit{u}_{\sf d}}
\newcommand{\paya}{\textit{u}_{\sf a}}
\newcommand{\pay}{\textit{u}}
\newcommand{\Payd}{\textit{U}_{\sf d}}
\newcommand{\Paya}{\textit{U}_{\sf a}}
\newcommand{\Pay}{\textit{U}}

\newcommand{\eqdef}{\ensuremath{\stackrel{\mathrm{def}}{=}}}
\newcommand{\argmax}{\operatornamewithlimits{argmax}}
\newcommand{\argmin}{\operatornamewithlimits{argmin}}
\newcommand{\supp}[1]{{\sf supp}(#1)}
\newcommand{\expect}{\operatornamewithlimits{\mathbb{E}}}
\newcommand{\expectDouble}[2]{\operatornamewithlimits{\displaystyle\mathbb{E}}_{\substack{#1\\ #2}}}

\renewcommand{\equiv}{\approx}

\newcommand{\vecC}{{\bm C}}
\newcommand{\ca}{c_{\sf a}}
\newcommand{\cd}{c_{\sf d}}
\newcommand{\Ca}{C_{\sf a}}
\newcommand{\Cd}{C_{\sf d}}
\newcommand{\sg}{\gamma}
\newcommand{\Sg}{\Gamma}

\newcommand{\vg}{V_{g}} 
\newcommand{\priorvg}[1]{\vg\left[#1\right]} 
\newcommand{\postvg}[2]{\vg\left[#1,#2\right]} 
\newcommand{\vf}{\mathbb{V}} 
\newcommand{\priorvf}[1]{\vf\left[#1\right]} 
\newcommand{\postvf}[2]{\vf\left[#1,#2\right]} 

\newcommand{\samplefrom}{\leftarrow} 


\newcommand{\operate}{\operatornamewithlimits{\mathbb{F}}}

\newcommand{\add}{\operatorname{+}}
\newcommand{\bigadd}{\operatorname{\sum}}
\newcommand{\hchoice}[1]{\;{{}_{#1}{\mathlarger{\mathlarger{\oplus}}}}\;}
\newcommand{\HChoice}[2]{{\osum_{#1 \samplefrom #2}}}
\newcommand{\HChoiceDouble}[4]{{\osum_{\substack{#1 \samplefrom #2\\ #3 \samplefrom #4}}}}
\newcommand{\hchoiceop}{\osum} 

\newcommand{\conc}{\diamond} 
\newcommand{\bigconc}{\meddiamond}
\newcommand{\vchoice}[1]{\;{{}_{#1}{\mathlarger{\mathlarger{\sqcupdot}}}}\;}
\newcommand{\VChoice}[2]{{\bigsqcupdot_{#1 \samplefrom #2}}}
\newcommand{\VChoiceDouble}[4]{{\bigsqcupdot_{\substack{#1 \samplefrom #2\\ #3 \samplefrom #4}}}}
\newcommand{\vchoiceop}{\bigsqcupdot} 

\newcommand{\Sqcup}{\operatornamewithlimits{\sqcup}}

\newcommand{\randarrow}{\mathbin{\stackrel{\$}{\leftarrow}}}
\newcommand{\proba}[1]{\mbox{\bf{Pr}}\left[#1\right]}
\def\dist#1{\mathbb{D}{#1}}

\newcommand{\qm}[1]{``#1''}
\newcommand{\true}{T}
\newcommand{\false}{F}


\section{Introduction}
\label{sec:introdction}
A fundamental problem in computer security is the leakage of sensitive information due to  
\emph{correlation} of secret values with 
\emph{observables}---i.e., any information accessible to the attacker, 
such as, for instance, the system's outputs or execution time.
The typical defense consists in reducing this correlation, 
which can be done in, essentially, two ways.
The first, applicable when the correspondence secret-observable is deterministic, 
consists in coarsening the equivalence classes of 
secrets that give rise to the same observables. 
This can be 
achieved with post-processing, i.e., sequentially composing the original system 
with a program that removes information from observables.
For example, a typical attack on encrypted web traffic consists in the analysis of the packets' 
length, and a typical defense consists in padding extra bits so to diminish the length variety 
\cite{sun:02:SandP}. 

The second kind of defense, on which we focus in this work, 
consists in adding controlled noise to the observables produced by
the system. 
This can be usually seen as a composition of different protocols via 
probabilistic choice.
\begin{Example}[Differential privacy]
Consider a counting query $f$, namely a function that, applied to a dataset $x$, returns the number of individuals in $x$ that satisfy a given property. A way to implement differential privacy \cite{Dwork:06:TCC} is to add geometrical noise to the result of $f$, so to obtain a probability distribution $P$ on integers of the form $P(z) = c \, e^{|z-f(x)|}$, where $c$ is a normalization factor. The resulting mechanism can  be interpreted as a probabilistic choice on protocols of the form $f(x), f(x) + 1, f(x) + 2, \ldots , f(x) - 1, f(x)-2, \ldots$, where the probability assigned to 
 $f(x) + n$ and to $f(x)-n$ decreases exponentially with $n$.
\end{Example}
\begin{Example}[Dining cryptographers]
Consider two agents running the dining cryptographers protocol \cite{Chaum:88:JC}, which consists in tossing a fair binary coin and then declaring 
the exclusive or $\oplus$ of their secret value $x$ and the result of the coin.  The protocol can be thought as the fair probabilistic choice  of two protocols, one consisting simply of declaring $x$, and the other declaring $x\oplus 1$. 
\end{Example}

Most of the work in the literature of quantitative information flow (QIF) considers passive attacks,
in which the attacker only observes the system. 
Notable exceptions are the works \cite{Boreale:15:LMCS,Mardziel:14:SP,Alvim:17:GameSec}, 
which consider attackers 
who interact with and influence the system, possibly in an adaptive way, with
the purpose of maximizing the leakage of information. 

\begin{Example}[CRIME attack]\label{exe:Crime}
Compression Ratio Info-leak Made Easy (CRIME) \cite{Rizzo:12:Ekoparty} is a
security exploit against secret web cookies over connections using the HTTPS and SPDY protocols and  data compression. 
The idea is that the attacker can inject some content $a$ in the communication of the secret $x$ from the target site to the server. The server then compresses and  encrypts the data, including both $a$ and $x$, and sends back the result. By observing the length of the result, the attacker can then infer information about $x$. 
To mitigate the leakage, one possible defense would consist in transmitting, along with $x$, also an encryption method $f$ selected randomly from a set $F$. Again, the resulting protocol can be seen as a composition, using probabilistic choice, of the protocols in the set $F$.  
\end{Example}

\begin{Example}[Timing side-channels]\label{ex:Timing}
Consider a password-checker, or any similar system in which the user authenticates himself
by entering a secret that is checked by the system. An adversary does not know the real secret,
of course, but a timing side-channel could reveal the \emph{part} (eg. which bit) of the secret in which the adversary's
input fails. By repeating the process with different inputs, the adversary might be able to
fully retrieve the secret. A possible counter measure is to make the side channel noisy,
by randomizing the order in which the secret's bit are checked against the user input.
This example is studied in detaul in Section~\ref{sec:password-example}.
\end{Example}

In all examples above the main use of the probabilistic choice is to  
obfuscate the relation between secrets and observables, 
thus reducing their correlation---and, hence, the information leakage. 
To achieve this goal, it is essential that the attacker never comes to 
know the result of the choice. 
In the CRIME example, however, if 
$f$ and  $a$ are chosen independently, then (in general) it is still better 
to choose $f$ probabilistically, even if the attacker will come to know, afterwards, 
the choice of $f$. 
In fact, this is true also for the attacker: his best strategies (in general) are to choose $a$
according to some probability distribution. 
Indeed, suppose that $F=\{f_1,f_2\}$ are the defender's choices and 
$A=\{a_1,a_2\}$ are the attacker's, and that 
$f_1(\cdot,a_1)$ leaks more than $f_1(\cdot,a_2)$, while $f_2(\cdot,a_1)$ 
leaks less than $f_2(\cdot,a_2)$. 
This is a scenario like \emph{the matching pennies} in game theory: 
if one player selects an action deterministically, the other player 
may exploit this choice and  get an advantage. 
For each player the optimal strategy is to play probabilistically, 
using a distribution that maximizes his own gain for all possible actions 
of the attacker. 
In zero-sum games, in which the gain of one player coincides 
with the loss of the other, the optimal pair of distributions always exists, 
and it is called \emph{saddle point}. 
It also coincides with the \emph{Nash equilibrium}, which is defined as the 
point in which neither of the two players gets any advantage in changing 
unilaterally his strategy. 

Motivated by these examples, this paper investigates the two kinds of choice, visible 
and hidden (to the attacker), in a game-theoretic setting. 
Looking at them as language operators, we study their algebraic properties, which will
help reason about their behavior in games. 
We consider zero-sum games, in which the gain (for the attacker) is represented by the leakage. 
While for visible choice it is appropriate to use the ``classic'' game-theoretic
framework, for hidden  choice we need to adopt the more general framework of the
\emph{information leakage games} proposed in \cite{Alvim:17:GameSec}. 
This happens because, in contrast with standard game theory, 
in games with hidden choice the payoff of a mixed strategy is 
a convex function of the distribution on the defender's pure actions, 
rather than simply the expected value of their utilities.
We will consider both simultaneous games---in which each player chooses independently---and sequential games--- in which one player chooses his action first. We aim at comparing all these situations, and at identifying the precise advantage of the hidden choice over the visible one. 

To measure leakage we use the well-known information-theoretic model. 
A central notion in this model is that of \emph{entropy}, but here
we  use its converse, \emph{vulnerability}, which 
represents  the magnitude of the threat. 
In order to derive results as general as possible, we  adopt 
the very comprehensive notion of vulnerability as any convex and continuous 
function, as  used in \cite{Boreale:15:LMCS} and 
\cite{Alvim:16:CSF}. 
This notion has been shown \cite{Alvim:16:CSF} 
to be, in a precise sense, the most general 
information measures w.r.t. a set of fundamental 
information-theoretic axioms.
Our results, hence, apply to all information measures that respect such
fundamental principles, including the widely adopted measures of \emph{Bayes vulnerability} (aka min-vulnerability, aka (the converse of) 
Bayes risk)~\cite{Smith:09:FOSSACS,Chatzikokolakis:08:JCS}, 
\emph{Shannon entropy}~\cite{Shannon:48:Bell}, 
\emph{guessing entropy}~\cite{Massey:94:IT}, and 
\emph{$g$-vulnerability}~\cite{Alvim:12:CSF}.
\\

The main contributions of this paper are:
\begin{itemize}
\item We present a general framework for reasoning about information leakage
in a game-theoretic setting, extending the notion of information leakage games proposed 
in~\cite{Alvim:17:GameSec} to both simultaneous and sequential games,
with either hidden or visible choice.

\item We present a rigorous compositional way, using visible and hidden choice operators, for representing attacker's and defender's actions in information leakage games.
In particular, we study the algebraic properties of visible and hidden choice on channels,
and compare the two kinds of choice with respect to the capability of reducing 
leakage, in  presence of an adaptive attacker.

\item 
We provide a taxonomy of the various scenarios (simultaneous and sequential) 
showing when randomization is necessary, for either attacker or defender, 
to achieve optimality.
Although it is well-known in information flow that the defender's best strategy 
is usually randomized, only recently it has been shown that
when defender and attacker act simultaneously, the attacker's optimal strategy 
also requires randomization~\cite{Alvim:17:GameSec}.

\item We compare the vulnerability of the leakage games for these various scenarios 
and establish a hierarchy 
of leakage games based on the order between the value of the leakage in the Nash equilibrium.
Furthermore, we show that when the attacker moves first in a sequential game with hidden choice, 
the behavioral strategies (where the defender chooses his probabilistic distribution after he has seen the choice of the attacker)  
are more advantageous for the defender than the mixed strategies (where the defender chooses the probabilistic distribution over his possible functional dependency on the choice of the attacker). This contrast 
with the standard game theory, where the two types of strategies are equivalent. 
Another difference is that in our attacker-first sequential games there may not exist Nash equilibria with deterministic strategies for the defender
(although the defender has full visibility of the attacker's moves).
\item We use our framework in a detailed case study of a password-checking protocol. 
A naive program, which checks the password bit by bit and stops when it finds a 
mismatch, is clearly very insecure, because it reveals at each attempt 
(via a timing side-channel), the maximum correct prefix. 
On the other hand, if we continue checking until the end of the string (time padding), 
the program becomes very inefficient. 
We show that, by using probabilistic choice instead, we can obtain a good trade-off 
between security and efficiency. 
\end{itemize}

\paragraph*{Plan of the paper.}
The remaining of the paper is organized as follows.
In Section~\ref{sec:preliminaries} we review some basic notions of game theory and 
quantitative information flow. 
In Section~\ref{sec:running-example} we introduce our running example. 
In Section~\ref{sec:operators} we define the visible and hidden choice operators 
and demonstrate their algebraic properties. 
In Section~\ref{sec:games-setup}, the core of the paper, we examine
various scenarios for leakage games. 
In Section~\ref{sec:comparing-games} we compare the vulnerability of the various leakage games,
and establish a hierarchy among those games.
In Section~\ref{sec:password-example} we show an application of our framework 
to a password checker. 
In Section~\ref{sec:related-work} we discuss related work and,
finally, in Section~\ref{sec:conclusion} we conclude.

A preliminary version of this paper appeared in~\cite{Alvim:18:POST}. 
One difference with respect to~\cite{Alvim:18:POST} is that in the present paper we consider both behavioral and mixed strategies in the sequential games, while in~\cite{Alvim:18:POST} we only considered the latter. We also show that the two kinds of strategies are not equivalent in our context (Example~\ref{eg:differ:mixed:behavioral}: the optimal strategy profile yields a different payoff depending on whether the defender adopts mixed strategies or behavioral ones). In light of this difference, we provide new results that concern behavioral  strategies, and in particular: 
\begin{itemize}
\item
Theorem~\ref{theo:deterministic-strategies-II}, which concerns the defender's behavioral  strategies in the defender-first game with visible choice (Game II) ,
\item
the second half of Theorem~\ref{Theo:deterministic solution Game VI}, which deals with  the adversary's behavioral strategies in the attacker-first game with hidden choice (Game VI).
\end{itemize}
Furthermore, in this paper we define formally all concepts, and provide all the proofs.
In particular, we provide a precise formulation of the comparison among games with visible/hidden choices (Propositions~\ref{prop:order:I-IV and III-VI} and~\ref{prop:mixed-behavioral}, Corollaries~\ref{theo:order:I-IV},~\ref{theo:order:III-VI}, and~\ref{cor:mixed-behavioral}) in Section~\ref{sec:comparing-games}.
Finally, in Section~\ref{sec:password-example} we provide a new result, expressed by  Theorem~\ref{prop:uniform-prior-delta},  regarding the optimal strategies for the defender in the presence of a uniform prior on passwords.

\section{Preliminaries}
\label{sec:preliminaries}
In this section we review some basic notions from game theory
and quantitative information flow.
We use the following notation:
Given a set $\cali$, we denote by $\distr\cali$ the 
\emph{set of all probability distributions} over $\cali$.
Given $\mu\in \distr\cali$, its \emph{support}
$\supp{\mu} \eqdef \{ i \in \cali : \mu(i)>0 \}$ 
is the set of its elements with positive probability.
We use $i{\samplefrom}\mu$ to indicate that a value 
$i{\in}\cali$ is sampled from a distribution $\mu$ on $\cali$.
A set $\cals\subseteq\reals^n$ is \emph{convex} if $ts_0 + (1-t)s_1 \in\cals$ for all
$s_0, s_1\in\cals$ and $t\in[0,1]$.
For such a set, a function $f:\cals\rightarrow\reals$ is \emph{convex} if 
$f(ts_0 + (1-t)s_1) \le tf(s_0) + (1-t)f(s_1)$
for all $s_0, s_1\in\cals,t\in[0,1]$, and
\emph{concave} if $-f$ is convex.

\subsection{Basic concepts from game theory}
\label{subsec:game-theory}
\subsubsection{Two-player games}

\emph{Two-player games} are a model for reasoning about 
the behavior of two players.
In a game, each player has at its disposal a set of \emph{actions} 
that he can perform, and he obtains some gain or loss depending on
the actions chosen by both  players.
Gains and  losses are defined using a real-valued \emph{payoff
function}.
Each player is assumed to be \emph{rational}, i.e.,  
his choice is driven by the attempt to maximize his own expected payoff.
We also assume that 
the set of possible actions and the payoff functions of both players 
are \emph{common knowledge}.

In this paper we only consider \emph{finite games}, in which the 
set of actions available to the players are finite, and which
are also \emph{zero-sum games}, so the payoff of one player is the loss of the other.
Next we introduce an important distinction between \emph{simultaneous} 
and \emph{sequential} games. 
In the following, we will call the two players \emph{defender} 
and \emph{attacker}.

\subsubsection{Simultaneous games}

In a  simultaneous game, each player chooses his action without knowing 
the action chosen by the other.
The term ``simultaneous'' here does  not  mean that the players' actions 
are chosen at the same time, but only that they are chosen independently.
Formally, such a game  is defined as a tuple\footnote{Following the convention of \emph{security games}, we set the first player to be the defender.} 
$(\cald, \cala, \payd, \paya)$, where
$\cald$ is a nonempty set of \emph{defender's actions}, 
$\cala$ is a nonempty set of \emph{attacker's actions},
$\payd: \cald\times\cala \rightarrow \reals$ is the \emph{defender's payoff function}, and
$\paya: \cald\times\cala \rightarrow \reals$ is the \emph{attacker's payoff function}.

Each player may choose an action deterministically or probabilistically.
A \emph{pure strategy} of the defender (resp. attacker) is a deterministic 
choice of an action, i.e., an element $d\in\cald$ (resp. $a\in\cala$). 
A pair $(d, a)$ is called \emph{pure strategy profile}, and 
$\payd(d, a)$, $\paya(d, a)$ represent the defender's and the attacker's 
payoffs, respectively. 
A \emph{mixed strategy} of the defender (resp. attacker) is a probabilistic 
choice of an action, defined as a probability distribution 
$\delta\in\distr\cald$ (resp. $\alpha\in\distr\cala$).
A pair  $(\delta, \alpha)$ is called  \emph{mixed strategy profile}.
The defender's and the attacker's \emph{expected payoff functions} for
mixed strategies are defined, respectively, as:
\begin{align*}
\Payd(\delta,\alpha)
&\eqdef {\expectDouble{d\leftarrow\delta}{a\leftarrow\alpha} \payd(d, a)}
=\sum_{\substack{d\in\cald\\ a\in\cala}} \delta(d) \alpha(a) \payd(d, a)
\mbox{ and }
\\[1ex]
\Paya(\delta,\alpha) 
&\eqdef {\expectDouble{d\leftarrow\delta}{a\leftarrow\alpha}  \paya(d, a)}
= \sum_{\substack{d\in\cald\\ a\in\cala}} \delta(d) \alpha(a) \paya(d, a)
{.}
\end{align*}

A defender's mixed strategy $\delta\in\distr\cald$ is a \emph{best response} to 
an attacker's mixed strategy $\alpha\in\distr\cala$ if 
$\Payd(\delta, \alpha) = \max_{\delta'\in\distr\cald}\Payd(\delta', \alpha)$.
Symmetrically, $\alpha\in\distr\cala$ is a \emph{best response} to 
$\delta\in\distr\cald$ if 
$\Paya(\delta, \alpha) = \max_{\alpha'\in\distr\cala}\Payd(\delta, \alpha')$.
A \emph{mixed-strategy Nash equilibrium} is a  profile $(\delta^*, \alpha^*)$ such 
that $\delta^*$ is a best response to $\alpha^*$ and vice versa. 
This means that in a Nash equilibrium, no unilateral deviation by any single 
player provides better payoff to that player.
If $\delta^*$ and $\alpha^*$ are point distributions 
concentrated on some $d^*\in\cald$ and $a^*\in\cala$ respectively, 
then $(\delta^*, \alpha^*)$ is a \emph{pure-strategy Nash equilibrium}, 
and will be denoted by $(d^*, a^*)$. 
While not all games have a pure strategy Nash equilibrium, every finite game 
has a mixed strategy Nash equilibrium.

\subsubsection{Sequential games}

In a sequential game players may take turns in choosing their
actions.
In this paper, we only consider the case in which each player 
moves only once, in such a way that one of the players (\emph{the leader}) 
chooses his action first, and commits to it, before the other player 
(\emph{the follower}) makes his choice. 
The follower may have total knowledge of the choice made by the leader, or only partial. 
We refer to the two scenarios by the terms \emph{perfect} and \emph{imperfect information}, 
respectively.
Another distinction is the kind of randomization used by the players, namely whether the follower chooses probabilistically his action after he knows (partially or totally) the move of the leader, or whether he chooses at the beginning of the game a probabilistic distribution on (deterministic) strategies that  depend on the (partial or total) knowledge of the move of the leader. In the first case the strategies are called  \emph{behavioral},  in the second case \emph{mixed}. 

We now give the precise definitions assuming that the leader is the defender.  
The definitions for the case in which the leader is the attacker are analogous. 

A \emph{defender-first sequential game with perfect information} is a 
tuple $(\cald,\, \allowbreak \calda,\, \allowbreak \payd, \allowbreak \paya)$ where $\cald$, $\cala$, $\payd$ and $\paya$ are
defined as in simultaneous games:
The choice of an action $d\in \cald$ represents a  pure strategy of the defender. As for the attacker, his choice $a\in\cala$ depends functionally on the prior choice $d$ of the defender, and for this reason the 
pure strategies of the attacker are functions $\psa:\calda$. 
As for the probabilistic strategies, those of the  defender are defined as in simultaneous games:  
namely they are distributions $\delta\in\distr\cald$. 
On the other hand, the attacker's probabilistic strategies can be defined in two different ways:
In the \emph{behavioral} case,  an attacker's probabilistic strategy is a function $\fsa: \cald \rightarrow \distr(\cala)$. Namely, the attacker 
chooses a distribution on its actions after he sees the move of the defender. 
In the \emph{mixed} case,  an attacker's probabilistic strategy is a probability distribution  $\msa\in\distr(\calda)$. Namely,   the attacker 
chooses \emph{a priori} a distribution on pure strategies.
The defender's and the attacker's \emph{expected payoff functions} for mixed strategies are defined, respectively, as
\begin{align*}
\mbox{Behavioral case}\;\;
&\left\{
\begin{array}{rclcl}
\Payd(\delta,\fsa) 
&\eqdef&{\displaystyle\expect_{d\leftarrow\delta}\,\expect_{a\leftarrow\fsa(d)}
  \payd(d, a)}
&=& \displaystyle\sum_{d\in\cald}\delta(d)\sum_{a\in\cala} \fsa(d) (a) \,\payd(d, a)
 \\[1ex]
\Paya(\delta,\fsa)
&\eqdef&{\displaystyle\expect_{d\leftarrow\delta}\,\expect_{a\leftarrow\fsa(d)}
  \paya(d, a)}
&=& \displaystyle\sum_{d\in\cald}\delta(d)\sum_{a\in\cala} \fsa(d) (a) \,\paya(d, a)
\end{array}
\right.
\\[2ex]
\mbox{Mixed case}\;\;
&\left\{
\begin{array}{rclcl}
\Payd(\delta,\msa) 
&\eqdef&{\displaystyle\expectDouble{d\leftarrow\delta}{\psa\leftarrow\msa}
  \payd(d, \psa(d))}
&=& \displaystyle\sum_{\substack{d\in\cald\\ \psa:\calda}} \delta(d) \msa(\psa) \payd(d, \psa(d))
 \\[1ex]
\Paya(\delta,\msa)
&\eqdef&{\displaystyle\expectDouble{d\leftarrow\delta}{\psa\leftarrow\msa} \paya(d, \psa(d))}
&=& \displaystyle\sum_{\substack{d\in\cald\\ \psa:\calda}} \delta(d) \msa(\psa) \paya(d, \psa(d))
\end{array}
\right.
\end{align*}

The case of imperfect information 
is typically formalized by assuming an \emph{indistinguishability (equivalence) relation} over the actions chosen 
by the leader, representing a scenario in which the follower cannot distinguish between the actions belonging to the same equivalence class. 
The pure strategies of the followers, therefore,  are functions from the set of the equivalence classes on the actions of the leader to his own actions. Formally,  a \emph{defender-first  sequential game with imperfect information } is a
tuple $(\cald,\, \infoseta\rightarrow\cala,\, \payd, \paya)$ where $\cald$, $\cala$, $\payd$ and $\paya$ are
defined as in simultaneous games, and $\infoseta$ is a partition of $\cald$.  
The \emph{expected payoff functions} are defined as before, except that now the argument of 
$\fsa$ and $\psa$ is the equivalence class of $d$. 
Note that in the case in which all defender's actions are indistinguishable from each other at the eyes of the attacker  (\emph{totally imperfect information}), we have  $\infoseta = \{\cald\}$ and the expected payoff
functions coincide with those of the simultaneous games. 
In contrast, in the games in which all defender's actions are distinguishable from the view of the attacker (\emph{perfect information}), we have $\infoseta = \{ \{ d \} \mid d\in\cald \}$.

In the standard game theory, under the assumption of \emph{perfect recall} (i.e., the players never forget what they have learned), behavioral and mixed strategies are equivalent, 
in the sense that  for any behavioral strategy there is a mixed strategy that yields the same payoff, and vice versa. This is true for both cases of perfect and imperfect information, see~\cite{Osborne:94:BOOK}, Chapter 11.4. 
In our leakage games, however, this equivalence does not hold anymore, as it will be shown in Sections~\ref{sec:games-setup} and \ref{sec:comparing-games}.

\subsubsection{Zero-sum games and Minimax Theorem}

A game $(\cald,\, \cala,\, \payd, \paya)$ is \emph{zero-sum} if for any $d\in\cald$ and any $a\in\cala$, the defender's loss is equivalent to the attacker's gain, i.e., $\payd(d, a) = -\paya(d, a)$.
For brevity, in zero-sum games we denote by $u$ the attacker's payoff function $\paya$, and by $U$ the attacker's expected payoff $\Paya$.\footnote{Conventionally in game theory  the payoff $u$ is  set to 
be that of the first player, but we prefer to look at the payoff from the point of view of the attacker to be in line with the definition of payoff as vulnerability.}
Consequently,  the goal of the defender is to minimize $U$, and the goal of the attacker is to maximize it. 

In simultaneous zero-sum games the Nash equilibrium corresponds to the solution of the \emph{minimax} problem (or equivalently,  the \emph{maximin} problem), namely, the strategy profile $(\delta^*, \alpha^*)$ such that 
$U(\delta^*, \alpha^*)=\min_{\delta} \max_{\alpha} U(\delta, \alpha)$. 
The  von Neumann's minimax theorem, in fact, ensures that such solution (which always exists) is stable.

\begin{restatable}[von Neumann's minimax theorem]{Theorem}{res:vonneumann}
\label{theo:vonneumann}
Let $\calx \subset \reals^m$ and $\caly \subset \reals^n$ be compact convex sets,
and $\Pay: \calx\times\caly\rightarrow\reals$ be a continuous function such that
$\Pay(x, y)$ is a convex function in $x\in\calx$ and a concave function in $y\in\caly$.
Then
\[
\min_{x\in\calx} \max_{y\in\caly} \Pay(x, y) = \max_{y\in\caly} \min_{x\in\calx} \Pay(x, y)
.
\]
\end{restatable}

A related property is that, under the conditions of Theorem~\ref{theo:vonneumann}, there exists a \emph{saddle point} $(x^*, y^*)$ s.t., for all $x{\in}\calx$  and $y{\in}\caly$: 
$
 \Pay(x^{*}, y) \leq \Pay(x^*, y^*) \leq \Pay(x, y^*) 
$.

The solution of the minimax problem can be obtained by using convex optimization techniques. 
In case $\Pay(x, y)$ is affine in $x$ and in $y$, we can also use linear optimization. 

In case $\cald$ and $\cala$ contain two elements each, there is a closed form for the solution.  
Let  $\cald =\{d_0,d_1\}$ and $\cala=\{a_0,a_1\}$ respectively. 
Let $u_{ij}$ be the payoff of the defender on $d_i, a_j$.           
Then the Nash equilibrium $(\delta^*,\alpha^*)$ is given by:
\begin{equation}
\label{eq:closedform}
\textstyle
\delta^*(d_0)=\frac{ u_{11} - u_{10}}{  u_{00} - u_{01} - u_{10} +u_{11}} \;\;\;\;
\alpha^*(a_0)=\frac{ u_{11} - u_{01}}{  u_{00} - u_{01} - u_{10} +u_{11}}
\end{equation}
if these values are in $[0,1]$.  
Note that, since there are only two elements, the strategy $\delta^*$ is completely specified by its value in $d_0$, and analogously for $\alpha^*$.

\subsection{Quantitative information flow}
\label{subsec:qif}

Finally, we briefly review the standard framework of
quantitative information flow, which is concerned with
measuring the amount of information leakage in a 
(computational) system.
  
\subsubsection{Secrets and vulnerability}
A \emph{secret} is some piece of sensitive information the 
defender wants to protect, such as a user's password, social 
security number, or current location. 
The attacker usually only has some partial knowledge 
about the value of a secret, represented as a probability
distribution on secrets called a \emph{prior}.
We denote by $\calx$ the set of possible secrets,
and we typically use $\pi$ to denote a prior belonging to 
the set $\dist{\calx}$ of probability distributions over 
$\calx$. 

The  \emph{vulnerability} of a secret is a measure of the payoff that it represents for the attacker. 
In this paper we consider a very general notion of  vulnerability, following~\cite{Alvim:16:CSF}, and we define 
a vulnerability $\vf$ to be any continuous and convex function of type $\dist{\calx} \rightarrow \reals$.
It has been shown in~\cite{Alvim:16:CSF} 
that these functions coincide
with the set of $g$-vulnerabilities, 
and are, in a precise sense, the most general 
information measures w.r.t. a set of fundamental
information-theoretic axioms.~\footnote{\review{More precisely, if posterior vulnerability 
is defined as the expectation of the vulnerability of posterior
distributions, 
the measure respects the fundamental information-theoretic properties of
\emph{data-processing inequality} (i.e., that post-processing can never increase information, but only destroy it), and of
\emph{non-negativity of leakage} (i.e., that by observing the output of a channel an actor cannot, in average, lose information) if, and only if, vulnerability is convex.}}
This notion, hence, subsumes all information measures that respect such
fundamental principles, including the widely adopted measures of \emph{Bayes vulnerability} (aka min-vulnerability, aka (the converse of) 
Bayes risk)~\cite{Smith:09:FOSSACS,Chatzikokolakis:08:JCS}, 
\emph{Shannon entropy}~\cite{Shannon:48:Bell}, 
\emph{guessing entropy}~\cite{Massey:94:IT}, and 
\emph{$g$-vulnerability}~\cite{Alvim:12:CSF}.

\subsubsection{Channels, posterior vulnerability, and leakage}
Computational systems can be modeled as information
theoretic channels.
A 
\emph{channel}
$C : \calx \times \caly \rightarrow \reals$ is a function
in which $\calx$ is a set of \emph{input values}, $\caly$ is a set 
of \emph{output values}, and $C(x,y)$ represents the conditional 
probability of the channel producing output $y \in \caly$ when 
input $x \in \calx$ is provided. 
Every channel $C$ satisfies $0 \leq C(x,y) \leq 1$ for all 
$x\in\calx$ and $y\in\caly$, and $\sum_{y\in\caly} C(x,y) = 1$ for all $x\in\calx$.

A distribution $\pi\in\dist{\calx}$ and a channel $C$ 
with inputs $\calx$ and outputs $\caly$ induce a joint distribution
$p(x,y) = \pi(x)C({x,y})$ on $\calx \times \caly$,
producing joint random variables $X, Y$ with marginal 
probabilities $p(x) = \sum_{y} p(x,y)$ and 
$p(y) = \sum_{x} p(x,y)$, and conditional probabilities 
$p(x{\mid}y) = \nicefrac{p(x,y)}{p(y)}$ if $p(y)\neq 0$. 
For a given $y$ (s.t. $p(y)\neq 0$), the conditional 
probabilities $p(x{\mid}y)$ for each $x \in \calx$ form the 
\emph{posterior distribution $p_{X \mid y}$}.

A channel $C$ in which $\calx$ is a set of secret values 
and $\caly$ is a set of observable values produced
by a system can be used to model computations on secrets.
Assuming the attacker has prior knowledge $\pi$ about
the secret value, knows how a channel $C$ works, and
can observe the channel's outputs, the effect of the channel 
is to update the attacker's knowledge from $\pi$ to a 
collection of posteriors $p_{X \mid y}$, each occurring 
with probability $p(y)$.%

Given a vulnerability $\vf$, a prior $\pi$, and a channel $C$, 
the \emph{posterior vulnerability} $\postvf{\pi}{C}$ is the vulnerability 
of the secret after the attacker has observed the output of the channel $C$.
Formally:
$\postvf{\pi}{C} \eqdef \sum_{y \in \caly} p(y) \priorvf{p_{X \mid y}}$.

Consider, for instance, the example of the password-checker with a timing side-channel
from the introduction (Example~\ref{ex:Timing}, also discussed in detail in
Section~\ref{sec:password-example}). Here the set of secrets $\calx$ consists of all possible
passwords (say, all strings of $n$ bits), and a natural vulnerability function is
Bayes-vulnerability, given by $\vf(\pi)=\max_{x\in\calx} \pi(x)$. This function expresses the adversary's probability
of \emph{guessing correctly} the password in one try; assuming that the passwords are chosen uniformly, i.e.
$\pi$ is uniform, any guess would be correct with
probability $2^{-n}$, giving $\vf(\pi) = 2^{-n}$.
Now, imagine that the timing side-channel reveals that the adversary's input failed
\emph{on the first bit}. The adversary now knows the first bit of the password (say $0$), 
hence the posterior $p_{X \mid y}$ assigns probability $0$ to all passwords with first bit
$1$, and probability $2^{-(n-1)}$ to all passwords with first bit $0$.
This happens for all possible posteriors, giving posterior vulnerability
$\postvf{\pi}{C}  = 2^{-(n-1)}$ ($2$ times greater than the prior $\vf$).

It is known from the literature~\cite{Alvim:16:CSF} 
that  the posterior vulnerability is a convex function of $\pi$. 
Namely, for any channel $C$, any family of distributions $\{\pi_i\}$, and any set of convex coefficients $\{c_i\}$, we have: 
\begin{align*}
\postvf{\sum_i c_i \pi_i}{C}
\leq &\, \sum_{i} c_i \postvf{\pi_i}{C}
\end{align*}

The \emph{(information) leakage} of 
channel $C$ under prior $\pi$ is a comparison between 
the vulnerability of the secret before the system
was run---called \emph{prior vulnerability}---and the 
posterior vulnerability of the secret.
Leakage reflects by how much the observation of 
the system's outputs increases the attacker's 
information about the secret. 
It can be defined either 
\emph{additively} ($\postvf{\pi}{C}-\priorvf{\pi}$), or
\emph{multiplicatively} ($\nicefrac{\postvf{\pi}{C}}{\priorvf{\pi}}$).
In the password-checker example, the additive leakage is
$2^{-(n-1)} - 2^{-n} = 2^{-n}$ and the multiplicative leakage is
$\nicefrac{2^{-(n-1)}}{2^{-n}} = 2$.

\section{An illustrative example}
\label{sec:running-example}
We introduce an example which will serve as running example through the paper. 
Although admittedly contrived, this example is simple and yet produces
different leakage measures for all different combinations of visible/hidden choice and simultaneous/sequential games, 
thus providing a way to compare all different scenarios we are interested in. 

\begin{figure}[tb]
\centering
\begin{minipage}[t]{0.3\linewidth}
\noindent 
\texttt{\underline{Program 0}}\\[1mm]
\noindent \texttt{\textbf{High Input:}} $x \in \{0,1\}$\\
\noindent \texttt{\textbf{Low Input:}} $a \in \{0,1\}$\\
\texttt{\textbf{Output:}} $y\in \{0,1\}$\\[0.5mm]
$y= x \cdot a$\\
\textbf{return} $y$
\end{minipage}
\hspace{0.15\linewidth}
\begin{minipage}[t]{0.3\linewidth}
\noindent 
\texttt{\underline{Program 1}}\\[1mm]
\noindent \texttt{\textbf{High Input:}} $x \in \{0,1\}$\\
\noindent \texttt{\textbf{Low Input:}} $a \in \{0,1\}$\\
\texttt{\textbf{Output:}} $y\in \{0,1\}$\\[0.5mm]
$c \samplefrom{\text{flip coin with bias $\nicefrac{a}{3}$}}$\\
\textbf{if} $c = \textit{heads}$ \{$y = x$\} \\
\textbf{else} \{$y = \overline{x}$\}\\
\textbf{return} $y$
\end{minipage}
\caption{Alternative programs for the running example.}
\label{fig:running-exa}
\end{figure}

Consider that a binary secret must be processed by a program.
As usual, a defender wants to protect the secret value, 
whereas an attacker wants to infer it by observing the 
system's output.
Assume the defender can choose which among two
alternative versions of the program to run.
Both programs take the secret value $x$ as high input,  
and a binary low input $a$ whose value is chosen by the attacker. 
They both return the output in a low variable $y$.~\footnote{We adopt the usual 
convention in QIF of referring to secret variables, inputs and outputs in programs 
as \emph{high}, and to their observable counterparts as \emph{low}.}
\texttt{Program 0} returns the binary product of $x$ and $a$,
whereas \texttt{Program 1} flips a coin with bias $\nicefrac{a}{3}$
(i.e., a coin which returns heads with probability $\nicefrac{a}{3}$)
and returns $x$ if the  result is
heads, and the complement $\overline{x}$ of $x$ otherwise.
The two programs are represented in Figure~\ref{fig:running-exa}.

The combined choices of the defender's and of the
attacker's determine how the system behaves.
Let $\cald{=} \{0,1\}$ represent the set of the defender's
choices---i.e., the index of the program to use---, and
$\cala = \{0,1\}$ represent the set of the attacker's
choices---i.e., the value of the low input $a$. 
We shall refer to the elements of $\cald$ and $\cala$ as \emph{actions}.
For each possible combination of actions 
$d \in \cald$ and $a \in \cala$, we can construct a channel 
$C_{da}$ modeling how the resulting system behaves.
Each channel $C_{da}$ is a function of type 
$\calx \times \caly \rightarrow \reals$, where 
$\calx = \{0,1\}$ is the set of possible high input values 
for the system, and $\caly = \{0,1\}$ is the set of possible 
output values from the system.
Intuitively, each channel provides the probability that the
system (which was fixed by the defender) produces output 
$y\in\caly$ given that the high input is $x \in \calx$
(and that the low input was fixed by the attacker). 
The four possible channels are depicted in Table~\ref{table:four-channels}.

\begin{table}[tb]
\centering
\vspace{-6mm}
\begin{subtable}[t]{0.45\linewidth}
\centering
\caption{The four channels $C_{da}$ for $d,a\in\{0,1\}$}
\begin{tabular}{cc}
$
\begin{array}{|c|c|c|}
\hline
C_{00} & y=0 & y=1 \\ \hline
x=0    & 1 & 0 \\
x=1    & 1 & 0 \\ \hline
\end{array}
$
&
$
\begin{array}{|c|c|c|}
\hline
C_{01} & y=0 & y=1 \\ \hline
x=0    & 1  & 0  \\
x=1    & 0  & 1  \\ \hline
\end{array}
$
\\
 & \\
$
\begin{array}{|c|c|c|}
\hline
C_{10} & y=0 & y=1 \\ \hline
x=0    & 0 & 1 \\
x=1    & 1 & 0 \\ \hline
\end{array}
$
&
$
\begin{array}{|c|c|c|}
\hline
C_{11} & y=0 & y=1 \\ \hline
x=0    & \nicefrac{1}{3} & \nicefrac{2}{3} \\
x=1    & \nicefrac{2}{3} & \nicefrac{1}{3} \\ \hline
\end{array}
$
\end{tabular}
\label{table:four-channels}
\end{subtable}
\qquad
\begin{subtable}[t]{0.45\linewidth}
\vspace{0.9cm}
\caption{Bayes vulnerability of each channel $C_{da}$.}
\centering
\begin{tabular}{c}
\renewcommand{\arraystretch}{1.15}
$
\begin{array}{|c|c|c|}
\hline
\vf & a=0 & a=1 \\ \hline
d=0    & \nicefrac{1}{2} & 1 \\ 
\hline
d=1    & 1 &  \nicefrac{2}{3}\\ \hline
\end{array}
$
\renewcommand{\arraystretch}{1}
\end{tabular}
\label{table:vulnerabilitygame}
\end{subtable}
\caption{Channel matrices and payoff table for running example.}
\label{table:runningexample}
\end{table}

Note that channel $C_{00}$ does not leak any 
information about the input $x$ (i.e., it is 
\emph{non-interferent}), whereas channels $C_{01}$ and 
$C_{10}$ completely reveal $x$.
Channel $C_{11}$ is an intermediate case: it leaks some 
information about $x$, but not all.

We want to investigate how the defender's and the attacker's choices
influence the leakage of the system. 
For that we can just consider the (simpler) notion of posterior vulnerability, 
since in order to make the comparison fair we need to assume that the prior is 
always the same in the various scenarios, and this implies that the 
leakage is in a one-to-one correspondence with the posterior vulnerability 
(this happens for both additive and multiplicative leakage). 

For this example, assume we are interested in 
Bayes vulnerability~\cite{Chatzikokolakis:08:JCS,Smith:09:FOSSACS}, defined as
$\vf(\pi) = \max_{x} \pi(x)$ for every $\pi \in \dist{\calx}$.
Assume for simplicity that 
the prior is the uniform prior $\pi_u$.
In this case we know from \cite{Braun:09:MFPS} that the posterior 
Bayes vulnerability of a channel is the sum of the greatest elements 
of each column, divided by the total number of inputs. 
Table~\ref{table:vulnerabilitygame} provides the
Bayes vulnerability $\vf_{da} \eqdef \postvf{\pi_{u}}{C_{da}}$
of each channel considered above.

Naturally, the attacker aims at maximizing the vulnerability of the system, 
while the defender tries to minimize it. 
The resulting vulnerability will depend on various factors, in particular on
whether the two players make their choice \emph{simultaneously} 
(i.e., without knowing the choice of the opponent) or \emph{sequentially}. 
Clearly, if the choice of a player who moves first is known by an opponent who moves second, 
the opponent will be in advantage. 
In the above example, for instance, if the defender knows the choice $a$ of the attacker, the most convenient choice for him is to set  $d=a$, and 
the vulnerability will be at most $\nicefrac{2}{3}$ . 
Vice versa, if the attacker knows the choice $d$  of the defender, the most convenient choice for him is to set  $a\neq d$. The vulnerability in this case will be $1$.

Things become more complicated when players make choices simultaneously. 
None of the pure choices of $d$ and $a$ are the best for the corresponding player, 
because the vulnerability of the system depends also on the (unknown) choice of the other player. 
Yet there is a strategy leading to the best possible situation for both players (the \emph{Nash equilibrium}), but it is mixed (i.e., probabilistic), in that the players randomize their choices according to some precise distribution. 

Another factor that affects vulnerability is whether or not the defender's choice is known to the attacker at the moment in which he observes the output of the channel. 
Obviously, this corresponds to whether or not the attacker knows what channel he is observing. 
Both cases are plausible: naturally the defender has all the interest in keeping his choice 
(and, hence, the channel used) secret, since then the attack will be less effective 
(i.e., leakage will be smaller). 
On the other hand, the attacker may be able to identify the channel used anyway, for instance 
because the two programs have different running times. 
We will call these two cases \emph{hidden} and \emph{visible} choice, respectively. 

It is possible to model players' strategies, as well as hidden and visible choices, 
as operations on channels. 
This means that we can look at the whole system as if it were a single channel, 
which will turn out to be useful for some proofs of our technical results.
Next section is dedicated to the definition of these operators. 
We will calculate the exact values for our example in Section~\ref{sec:games-setup}.

\section{Choice operators for protocol composition}
\label{sec:operators}
In this section we define 
the operators of visible and hidden choice for protocol 
composition. These operators are formally defined on the channel matrices of the protocols, and, 
since  channels are a particular kind of 
matrix, we use these matrix operations to define
the operations of visible and hidden choice among
channels, and to prove important properties of 
these channel operations.

\subsection{Matrices, and their basic operators}
\label{sec:matrices-basics}

Given two  sets $\calx$ and $\caly$, a \emph{matrix} 
is a total function of type $\calx \times \caly \rightarrow \reals$.
Two matrices 
$M_{1}: \calx_{1} \times \caly_{1} \rightarrow \reals$ and 
$M_{2}: \calx_{2} \times \caly_{2} \rightarrow \reals$ 
are said to be \emph{compatible} if $\calx_{1} = \calx_{2}$.
If it is also the case that $\caly_{1} = \caly_{2}$, we say 
that the matrices \emph{have the same type}.
The \emph{scalar multiplication} $r{\cdot}M$ between a scalar $r$ and
a matrix $M$ is defined as usual, and so is the \emph{summation}
$\left(\sum_{i \in \cali} M_{i}\right)(x,y) = M_{i_{1}}(x,y) \add \ldots \add M_{i_{n}}(x,y)$
of a family $\{M_{i}\}_{i \in \cali}$ of matrices all of a same type.

Given a family $\{M_{i}\}_{i \in \cali}$ of
compatible matrices s.t. each $M_{i}$ has type 
$\calx \times \caly_{i} \rightarrow \reals$, their \emph{concatenation} $\bigconc_{i \in \cali}$ 
is the matrix having all columns of every matrix in the family,
in such a way that every column is tagged with the matrix it
came from.
Formally, 
$\left( \bigconc_{i \in \cali} M_{i} \right)(x,(y,j)) = \,M_{j}(x,y)$, 
if $y \in \caly_{j}$,
and the resulting matrix has type
$\calx \times \left( \bigsqcup_{i \in \cali} \caly_{i} \right) \rightarrow \reals$.~\footnote{%
$\bigsqcup_{i \in \cali} \caly_{i} = \caly_{i_{1}} \sqcup \caly_{i_{2}} \sqcup \ldots \sqcup \caly_{i_{n}}$ denotes the \emph{disjoint union} 
$\{ (y,i) \mid y \in \caly_{i}, i \in \cali \}$
of the sets $\caly_{i_{1}}$, $\caly_{i_{2}}$, $\ldots$, $\caly_{i_{n}}$.
}
When the family $\{M_{i}\}$ has only two elements
we may use the \emph{binary} version $\conc$ of the
concatenation operator.
The following 
depicts the concatenation of two matrices $M_{1}$ and $M_{2}$ in tabular form.
$$
\begin{array}{|c|cc|}
\hline
M_{1} & y_{1} & y_{2} \\ \hline
x_{1} & 1 & 2 \\
x_{2} & 3 & 4 \\ \hline
\end{array}
\conc
\begin{array}{|c|ccc|}
\hline
M_{2} & y_{1} & y_{2} & y_{3} \\ \hline
x_{1} & 5 & 6 & 7 \\ 
x_{2} & 8 & 9 & 10 \\ \hline
\end{array} 
\;=\;
\begin{array}{|c|ccccc|}
\hline
M_{1} \conc M_{2} & (y_{1},1) & (y_{2},1) & (y_{1},2) & (y_{2},2) & (y_{3},2) \\ \hline
x_{1} & 1 & 2 & 5 & 6 & 7 \\ 
x_{2} & 3 & 4 & 8 & 9 & 10 \\ \hline
\end{array}
$$

\subsection{Channels, and their hidden and visible choice operators}
\label{sec:operators-definition}

A channel 
is a \emph{stochastic} matrix, i.e., 
all elements are non-negative, and all rows sum up to 1.
Here we will define two operators specific for channels.
In the following, for any real value $0 \leq p \leq 1$, 
we denote by $\overline{p}$ the value $1 - p$.

\subsubsection{Hidden choice}
The first operator models a hidden probabilistic 
choice among channels.
Consider a family $\left\{ C_{i} \right\}_{i \in \cali}$ 
of channels of a same type.
Let $\mu \in \dist{\cali}$ be a probability
distribution on the elements of the index set $\cali$.
Consider an input $x$ is fed to one of
the channels in $\left\{ C_{i} \right\}_{i \in \cali}$,
where the channel is randomly picked according to $\mu$.
More precisely, an index $i \in \cali$ is sampled with 
probability $\mu(i)$,  then the input
$x$ is fed to channel $C_{i}$,
and the output $y$ produced by the channel is then made visible, 
but not the index $i$ of the channel that was used.
Note that we consider hidden choice only among
channels of a same type:  if the sets of outputs
were not identical, the produced output 
might implicitly reveal which channel was used.

Formally, given a family $\{C_{i}\}_{i \in \cali}$ of 
channels s.t. each $C_{i}$ has same type 
$\calx \times \caly \rightarrow \reals$, the 
\emph{hidden choice operator} $\HChoice{i}{\mu}$ 
is defined as 
$
\HChoice{i}{\mu} C_{i} = \sum_{i \in \cali} \mu(i) \,C_{i}
$.

\begin{restatable}[Type of hidden choice]{prop}{restypehiddenchoice}
\label{prop:type-hidden-choice}
Given a family $\{C_{i}\}_{i \in \cali}$ of 
channels of  type $\calx \times \caly \rightarrow \reals$, 
and a distribution $\mu$ on $\cali$, 
the hidden choice 
$\HChoice{i}{\mu} C_{i}$ is 
a channel of type
$\calx \times \caly \rightarrow \reals$.
\end{restatable}
See Appendix~\ref{sec:proofs} for the proof.

In the particular case in which the family $\{C_{i}\}$ 
has only two elements $C_{i_{1}}$ and $C_{i_{2}}$, the
distribution $\mu$ on indexes is completely determined 
by a real value $0 \leq p \leq 1$ s.t. 
$\mu(i_{1}) = p$ and $\mu(i_{2}) = \overline{p}$.
In this case we may use the \emph{binary} version $\hchoice{p}$
of the hidden choice operator:
$
C_{i_{1}}{\hchoice{p}}C_{i_{2}} = p\,C_{i_{1}}{+}\overline{p}\,C_{i_{2}}
$.
%
The example 
below depicts the hidden choice
between channels $C_{1}$ and $C_{2}$, with probability
$p{=}\nicefrac{1}{3}$.
$$
\begin{array}{|c|cc|}
\hline
C_{1} & y_{1} & y_{2} \\ \hline
x_{1} & \nicefrac{1}{2} & \nicefrac{1}{2} \\
x_{2} & \nicefrac{1}{3} & \nicefrac{2}{3} \\ \hline
\end{array}
\hchoice{\nicefrac{1}{3}}
\begin{array}{|c|cc|}
\hline
C_{2} & y_{1} & y_{2} \\ \hline
x_{1} & \nicefrac{1}{3} & \nicefrac{2}{3} \\
x_{2} & \nicefrac{1}{2} & \nicefrac{1}{2} \\ \hline
\end{array}
\; = \;
\begin{array}{|c|cc|}
\hline
C_{1} \hchoice{\nicefrac{1}{3}} C_{2} & y_{1} & y_{2} \\ \hline
x_{1} & \nicefrac{7}{18} & \nicefrac{11}{18} \\ 
x_{2} & \nicefrac{4}{9} & \nicefrac{5}{9} \\ \hline
\end{array}
$$

\subsubsection{Visible choice}
The second operator models a visible probabilistic 
choice among channels.
Consider a family $\left\{ C_{i} \right\}_{i \in \cali}$ 
of compatible channels.
Let $\mu \in \dist{\cali}$ be a probability
distribution on the elements of the index set $\cali$.
Consider an input $x$ is fed to one of
the channels in $\left\{ C_{i} \right\}_{i \in \cali}$,
where the channel is randomly picked according to $\mu$.
More precisely, an index $i \in \cali$ is sampled with 
probability $\mu(i)$,  then the input
$x$ is fed to channel $C_{i}$, and the output $y$ produced by the channel is then made visible, 
along with the index $i$ of the channel that was used.
Note that visible choice makes sense only 
between compatible channels,
but it is not required that the output set of each channel 
be the same.

Formally, given  $\{C_{i}\}_{i \in \cali}$ of
compatible channels s.t. each $C_{i}$ has type 
$\calx \times \caly_{i} \rightarrow \reals$, 
and a distribution $\mu$ on $\cali$, the 
\emph{visible choice operator} $\VChoice{i}{\mu}$ 
is defined as 
$
\VChoice{i}{\mu} C_{i} = \bigconc_{i \in \cali} \;\mu(i)\, C_{i}
$.

\begin{restatable}[Type of visible choice]{prop}{restypevisiblechoice}
\label{prop:type-visible-choice}
Given a family $\{C_{i}\}_{i \in \cali}$ of 
compatible channels s.t. each $C_{i}$ has type 
$\calx \times \caly_{i} \rightarrow \reals$,
and a distribution $\mu$ on $\cali$, 
the result of the visible choice 
$\VChoice{i}{\mu} C_{i}$ is a channel
of type 
$\calx \times \left( \bigsqcup_{i \in \cali} \caly_{i} \right) \rightarrow \reals$.
\end{restatable}
See Appendix~\ref{sec:proofs} for the proof.

In the particular case the family $\{C_{i}\}$ 
has only two elements $C_{i_{1}}$ and $C_{i_{2}}$, the
distribution $\mu$ on indexes is completely determined 
by a real value $0 \leq p \leq 1$ s.t. 
$\mu(i_{1}) = p$ and $\mu(i_{2}) = \overline{p}$.
In this case we may use the \emph{binary} version $\vchoice{p}$
of the visible choice operator:
$
C_{i_{1}} \vchoice{p} C_{i_{2}} = p\, C_{i_{1}} \conc \overline{p}\, C_{i_{2}}
$.
%
The following 
depicts the visible choice 
betwee channels $C_{1}$ and $C_{3}$, with probability
$p{=}\nicefrac{1}{3}$.
$$
\begin{array}{|c|cc|}
\hline
C_{1} & y_{1} & y_{2} \\ \hline
x_{1} & \nicefrac{1}{2} & \nicefrac{1}{2} \\
x_{2} & \nicefrac{1}{3} & \nicefrac{2}{3} \\ \hline
\end{array}
\vchoice{\nicefrac{1}{3}}
\begin{array}{|c|cc|}
\hline
C_{3} & y_{1} & y_{3} \\ \hline
x_{1} & \nicefrac{1}{3} & \nicefrac{2}{3} \\
x_{2} & \nicefrac{1}{2} & \nicefrac{1}{2} \\ \hline
\end{array}
\; = \;
\begin{array}{|c|cccc|}
\hline
C_{1} \vchoice{\nicefrac{1}{3}} C_{3} & (y_{1},1) & (y_{2},1) & (y_{1},3) & (y_{3},3) \\ \hline
x_{1} & \nicefrac{1}{6} & \nicefrac{1}{6} & \nicefrac{2}{9} & \nicefrac{4}{9} \\ 
x_{2} & \nicefrac{1}{9} & \nicefrac{2}{9} & \nicefrac{1}{3} & \nicefrac{1}{3} \\ \hline
\end{array}
$$

\subsection{Properties of hidden and visible choice operators}
\label{sec:operators-properties}

We now prove algebraic properties of channel operators.
These properties will be useful when we model a (more complex) 
protocol as the composition of smaller channels via hidden or visible 
choice.

Whereas the properties of hidden choice 
hold generally with equality, those of 
visible choice are subtler.
For instance, visible choice is not idempotent,
since in general $C{\vchoice{p}}C \neq C$.
(In fact if $C$ has type $\calx \times \caly \rightarrow \reals$,
$C{\vchoice{p}}C$ has type $\calx \times (\caly \sqcup \caly) \rightarrow \reals$.)
However, idempotency and other properties involving visible
choice hold if we replace the notion of 
equality with the more relaxed notion of \qm{equivalence} 
between channels.
Intuitively, two channels are equivalent if they
have the same input space and yield 
the same value of vulnerability for every 
prior and every vulnerability function.

\begin{Definition}[Equivalence of channels]
\label{def:equivalence-channels}
Two compatible channels $C_{1}$ and $C_{2}$
with domain $\calx$
are \emph{equivalent}, denoted by $C_{1} \equiv C_{2}$,
if for every prior $\pi \in \dist{\calx}$ and every 
posterior vulnerability $\vf$ we have
$
\postvf{\pi}{C_{1}} = \postvf{\pi}{C_{2}}
$.
\end{Definition}

Two equivalent channels are indistinguishable from the point of view of information leakage, and in most cases we can just identify them. 
Indeed, nowadays there is a tendency to use \emph{abstract channels}~\cite{McIver:14:POST,Alvim:16:CSF}, which capture exactly 
the important behavior with respect to any  form of leakage. 
In this paper, however, we cannot use abstract channels because the hidden choice operator needs a concrete representation in order to be defined
unambiguously. 

The first properties we prove regard idempotency of operators, which
can be used do simplify the representation of some protocols.

\begin{restatable}[Idempotency]{prop}{residempotency}
\label{prop:idempotency}
Given a family $\{C_{i}\}_{i \in \cali}$ of 
channels s.t. $C_{i} = C$ for all $i \in \cali$,
and a distribution $\mu$ on $\cali$, then:
(a) ${\HChoice{i}{\mu}C_{i} = C}$; and
(b) ${\VChoice{i}{\mu}C_{i} \equiv C}$.
\end{restatable}
See Appendix~\ref{sec:proofs} for the proof.

The following properties regard the reorganization of operators, 
and they will be essential in some technical results in which we invert 
the order in which hidden and visible choice are applied in a protocol.

\begin{restatable}[\qm{Reorganization of operators}]{prop}{resreorganizationoperators}
\label{prop:reorganization-operators}
Given a family $\{C_{i j}\}_{i \in \cali, j \in \calj}$ of 
channels indexed by sets $\cali$ and $\calj$,
a distribution $\mu$ on $\cali$, and
a distribution $\eta$ on $\calj$:

\begin{enumerate}
\renewcommand{\labelenumi}{(\alph{enumi})}
\item 
${\HChoice{i}{\mu} \;\HChoice{j}{\eta} C_{i j} = \HChoiceDouble{i}{\mu}{j}{\eta} C_{i j}}$, if all $C_{i}$'s have the same type;\\

\item 
${\VChoice{i}{\mu} \; \VChoice{j}{\eta} C_{i j} \equiv \VChoiceDouble{i}{\mu}{j}{\eta} C_{i j}}$, if all $C_{i}$'s are compatible; and\\

\item \label{item:c}
${\HChoice{i}{\mu} \; \VChoice{j}{\eta} C_{i j} \equiv \VChoice{j}{\eta}\; \HChoice{i}{\mu} C_{i j}}$, 
if, for each $i$, all $C_{i j}$'s have same type $\calx \times \caly_{j}\rightarrow\reals$.
\end{enumerate}
\end{restatable}
See Appendix~\ref{sec:proofs} for the proof.

Finally, analogous properties of the binary operators are shown in Appendix~\ref{sec:operators-properties-binary}.

\subsection{Properties of vulnerability w.r.t. channel operators}
\label{sec:convexity-vulnerability}

We now derive some relevant properties of vulnerability w.r.t. our 
channel operators, which will be later used to obtain the Nash equilibria 
in information leakage games with different choice operations.

The first result states that posterior vulnerability is 
convex w.r.t. hidden choice (this result was already presented in \cite{Alvim:17:GameSec}), and 
linear w.r.t. to visible choice.

\begin{restatable}[Convexity/linearity of posterior vulnerability w.r.t. choices]{Theorem}{resconvexVq}
\label{theo:convex-V-q}
Let $\{C_{i}\}_{i \in \cali}$ be a family of channels,
and $\mu$ be a distribution on $\cali$.
Then, for every distribution $\pi$ on $\calx$, and every
vulnerability $\vf$:
\begin{enumerate}
\item \label{enumerate:vg:h:convex}
posterior vulnerability is convex w.r.t. to hidden choice:
$
\postvf{\pi}{\HChoice{i}{\mu} C_{i}} \leq \sum_{i \in \cali} \mu(i) \,\postvf{\pi}{C_{i}}
$
if all $C_{i}$'s have the same type.

\item \label{enumerate:vg:v:linear}
posterior vulnerability is linear w.r.t. to visible choice:
$
\postvf{\pi}{\VChoice{i}{\mu} C_{i}} = \sum_{i \in \cali} \mu(i) \,\postvf{\pi}{C_{i}}
$
if all $C_{i}$'s are compatible.
\end{enumerate}
\end{restatable}

\begin{proof}
\begin{enumerate}

\item Let us call $\calx \times \caly \rightarrow \reals$ the 
type of each channel $C_{i}$ in the family $\{C_{i}\}$.
Then:
\begin{align*}
\postvf{\pi}{\HChoice{i}{\mu} C_{i}} 
=&\, \postvf{\pi}{\bigadd_{i} \mu(i)C_{i}} & \text{(by definition of hidden choice)} \\ 
=&\, \sum_{y \in \caly} p(y) \cdot \priorvf{\frac{\pi(\cdot) \sum_{i}\mu(i)C_{i}(\cdot,y) }{p(y)}}  & \text{(by definition of posterior $\vf$)}\\ 
=&\, \sum_{y \in \caly} p(y) \cdot \priorvf{\sum_{i}\mu(i)\frac{\pi(\cdot) C_{i}(\cdot,y) }{p(y)}}  & \text{}\\ 
\leq&\, \sum_{y \in \caly} p(y) \cdot \sum_{i} \mu(i) \priorvf{\frac{\pi(\cdot) C_{i}(\cdot,y) }{p(y)}}  & \text{(by convexity of $\vf$)}\\ 
=&\, \sum_{i} \mu(i) \sum_{y \in \caly} p(y) \priorvf{\frac{\pi(\cdot) C_{i}(\cdot,y) }{p(y)}}  & \text{}\\
=&\, \sum_{i} \mu(i) \postvf{\pi}{C_{i}}  & \text{}
\end{align*}
where $p(y) = \sum_{x\in\calx} \pi(x) \sum_{i} \mu(i)C_{i}(x,y)$.
\\

\item Let us call $\calx \times \caly_{i} \rightarrow \reals$ the 
type of each channel $C_{i}$ in the family $\{C_{i}\}$.
Then:
\begin{align*}
\postvf{\pi}{\VChoice{i}{\mu} C_{i}}
=&\, \postvf{\pi}{\bigconc_{i} \mu(i)C_{i}} & \text{(by definition of visible choice)} \\ 
=&\, \sum_{y \in \caly} p(y) \cdot \priorvf{\frac{\pi(\cdot) \bigconc_{i}\mu(i)C_{i}(\cdot,y) }{p(y)}}  & \text{(by definition of posterior $\vf$)}\\ 
=&\, \sum_{y \in \caly} p(y) \cdot \priorvf{\bigconc_{i}\mu(i)\frac{\pi(\cdot) C_{i}(\cdot,y) }{p(y)}}  & \text{}\\ 
=&\, \sum_{y \in \caly} p(y) \cdot \sum_{i} \mu(i) \priorvf{\frac{\pi(\cdot) C_{i}(\cdot,y) }{p(y)}}  & \text{(see (*) below)}\\ 
=&\, \sum_{i} \mu(i) \sum_{y \in \caly} p(y) \priorvf{\frac{\pi(\cdot) C_{i}(\cdot,y) }{p(y)}}  & \text{}\\
=&\, \sum_{i} \mu(i) \postvf{\pi}{C_{i}}  & \text{}
\end{align*}
where $p(y) = \sum_{x\in\calx} \pi(x) \sum_{i} \mu(i)C_{i}(x,y)$, and 
step (*) holds because in the vulnerability of a concatenation of matrices
every column will contribute to the vulnerability in proportion to its weight
in the concatenation, and hence it is possible to break the vulnerability of a 
concatenated matrix as the weighted sum of the vulnerabilities of its sub-matrices.
\end{enumerate}
\end{proof}

The next result is concerned with posterior vulnerability 
under the composition of channels using both operators.

\begin{restatable}[Convex-linear payoff function]{Corollary}{resconcaveconvexV}
\label{cor:concave-convex-V}
Let $\{C_{i j}\}_{i \in \cali, j \in \calj}$ be a family of channels,
all with domain $\calx$ and with the same type, and let $\pi\in \distr\calx$, and $\vf$ be any vulnerability. 
Define
$\Pay: \distr\cali\times\distr\calj\rightarrow \reals$ as follows:
$
\Pay(\mu, \eta) \eqdef \postvf{\pi}{\HChoice{i}{\mu} \; \VChoice{j}{\eta} \; C_{i j}}
$.
Then $\Pay$ is convex on $\mu$ and linear on $\eta$.
\end{restatable}

\begin{proof}
To see that $U(\mu,\eta)$ is convex on $\mu$, note that:
\begin{align*}
U(\mu,\eta) 
=&\, \postvf{\pi}{\HChoice{i}{\mu} \; \VChoice{j}{\eta} \; C_{i j}} & \text{(by definition)} \\
\leq&\, \sum_{i} \mu(i) \; \postvf{\pi}{\VChoice{j}{\eta} \; C_{i j}} & \text{(by Theorem~\ref{theo:convex-V-q})}
\end{align*}

To see that $U(\mu,\eta)$ is linear on $\eta$, note that:
\begin{align*}
U(\mu,\eta) 
=&\, \postvf{\pi}{\HChoice{i}{\mu} \; \VChoice{j}{\eta} \; C_{i j}} & \text{(by definition)} \\
=&\, \postvf{\pi}{\VChoice{j}{\eta} \; \HChoice{i}{\mu} \; C_{i j}} & \text{(by Prop.~\ref{prop:reorganization-operators})} \\
=&\, \sum_{j} \eta(j) \; \postvf{\pi}{\HChoice{i}{\mu} \; C_{i j}} & \text{(by Theorem.~\ref{theo:convex-V-q})}
\end{align*}
~
\end{proof}

\section{Information leakage games}
\label{sec:games-setup}
In this section we present our framework for  reasoning about information leakage, 
extending the notion of \emph{information leakage games} proposed in \cite{Alvim:17:GameSec} from only simultaneous games with hidden choice to 
both simultaneous and sequential games, with either hidden or visible choice.

In an information leakage game the defender tries to minimize the leakage of information 
from the system, while the attacker tries to maximize it.
In  this basic scenario, their goals are just opposite (zero-sum). 
Both of them can influence the execution and the observable behavior of the system via a 
specific set of actions.
We assume players to be rational (i.e., they are able to figure out what is 
the best strategy to 
maximize their expected payoff), and that the set of actions and the
payoff function are common knowledge. 

Players choose their own strategy, which in general may be 
probabilistic (i.e., behavioral or mixed)
and choose their action by a random draw 
according to that strategy. 
After both players have performed their actions, the system runs 
and produces some output value which is visible to the attacker and
may leak some information about the secret.
The amount of leakage constitutes the attacker's gain, and the defender's loss.

To quantify the  leakage we model the system as an information-theoretic channel  
(cf. Section~\ref{subsec:qif}). We recall that leakage is defined as the difference (additive leakage) 
or the ratio (multiplicative leakage) between posterior and prior vulnerability. 
Since we are only interested in comparing the leakage of different channels
for a given prior, 
\emph{we will define the payoff just as the posterior vulnerability},
as the value of prior vulnerability will be the same for every channel.

\subsection{Defining information leakage games}
\label{subsec:def:leakage-game}

An (\emph{information}) \emph{leakage game}
consists of:
\begin{itemize}
\item[(1)]
two nonempty sets $\cald$, $\cala$ of defender's and attacker's actions,  
respectively,
\item[(2)]
a function $C: \cald \times \cala \rightarrow (\calx \times \caly \rightarrow \reals)$ 
that associates to each pair of actions  $(d,a)\in\cald \times \cala$ a channel 
$C_{da}: \calx \times \caly \rightarrow \reals$, 
\item[(3)]
a prior $\pi \in \dist\calx$ on secrets, and 
\item[(4)]
a vulnerability measure $\vf$, used to define the payoff function $u: \cald \times \cala \rightarrow\reals$ for pure strategies as
$
u(d,a) \eqdef \postvf{\pi}{C_{da}}
$.
We have only one payoff function because the game is zero-sum.
\end{itemize}
\noindent

Like in traditional game theory, the order of actions and the extent by which a player knows the move performed by the opponent play a critical role in deciding strategies and determining the payoff. 
In security, however, knowledge of the opponent's move affects the game in yet another way:
the effectiveness of the attack, i.e., the amount of leakage, depends crucially on whether or not
the attacker knows what channel is being used. 
It is therefore convenient to distinguish two phases in the leakage game: 

\begin{description}
\item[\textbf{Phase 1 - Determination of players' strategies, and the subsequent choice of their actions}:] Each player determines the most convenient strategy (which in general is probabilistic) for himself, and draws his action accordingly. One of the players may commit first to his action, and his choice may or may not be revealed to the follower. In general, knowledge of the leader's action may help the follower choose a more advantageous strategy. 

\item[\textbf{Phase 2 - Observation of the resulting channel's output, and payoff computation}:]  The attacker observes the output of the selected channel $C_{da}$ and performs his attack on the secret. In case he knows the defender's action, he is able to determine the exact channel $C_{da}$ being used (since, of course, the attacker knows his own action), 
and his payoff will be the posterior vulnerability $\postvf{\pi}{C_{da}}$.
However, if the attacker does not know exactly which channel has been used, then his payoff will be smaller.  
\end{description}

Note that the issues raised in Phase 2 are typical of leakage games; 
they do not have a correspondence (to the best of our knowledge) 
in traditional game theory. 
Indeed, in traditional game theory the resulting payoff is 
a deterministic function of all players' actions.
On the other hand, 
the extra level of randomization provided by the channel
is central to security, as 
it reflects
the principle of preventing the attacker from inferring 
the secret by obfuscating the link between secret and observables. 

Following the above discussion, we consider various possible 
scenarios for games, along two lines of classification.
The first classification concerns Phase 1 of the game, in which 
strategies are selected and actions are drawn, and consists
in three possible orders for the two players' actions.

\begin{description}
\item[{\bf Simultaneous:}]
The players choose (draw) their actions in parallel, each without knowing the choice of the other.
\item[{\bf Sequential, defender-first:}]
The defender draws an action, and commits to it,  before the attacker does.
\item[{\bf Sequential, attacker-first:}]
The attacker draws an action, and commits to it,  before the defender does.
\end{description}
Note that these sequential games may present imperfect information
(i.e., the follower may not know the leader's action), 
and that we have to further specify whether we use 
behavioral or mixed strategies.

The second classification concerns Phase 2 of the game, in which 
some leakage occurs in consequence of the attacker's observation of the channel's
output, and consists in two kinds of knowledge the attacker may have at this 
point about the channel that was used.
\begin{description}
\item[{\bf Visible choice:}] 
The attacker knows the defender's action when he observes the output of the channel, and therefore 
he knows which channel is being used.
Visible choice is modeled by the operator~$\vchoiceop$.  
\item[{\bf Hidden choice:}]
The attacker does not know the defender's action when he observes the output of the channel, and therefore in general he does not exactly know which channel is used (although in some special cases he may infer it from the output).  
Hidden choice is modeled by the operator~$\hchoiceop$. 
\end{description}

Note that the distinction between sequential and simultaneous games
is orthogonal to that between visible and hidden choice.
Sequential and simultaneous games model whether or not, respectively,
the follower's choice can be affected by knowledge of the leader's action.
This dichotomy captures how knowledge about the other player's actions 
can \emph{help a player choose his own action}, 
and it concerns how Phase 1 of the game occurs.
On the other hand, visible and hidden choice capture whether or not, respectively, the 
attacker is able to fully determine the channel representing the system, 
\emph{once the defender and attacker's actions have already been fixed}.
This dichotomy reflects the different \emph{amounts of information leaked} by 
the system as viewed by the attacker, 
and it concerns how Phase 2 of the game occurs.
For instance, in a simultaneous game neither player can choose his action based on 
the choice of the other.
However, depending on whether or not the defender's choice is visible, 
the attacker will or will not, respectively, be able to completely recover 
the channel used, which will affect the amount of leakage.

If we consider also the subdivision of sequential games into perfect and imperfect information, 
there are $10$ possible different combinations. 
Some, however, make little sense. 
For instance, the defender-first sequential game with perfect information 
(by the attacker) 
does not combine naturally with hidden choice
$\hchoiceop$, 
since that would mean that the attacker knows the action of the defender 
and chooses his strategy accordingly, but forgets it at the moment of computing the channel and its vulnerability.
(We assume \emph{perfect recall}, i.e., the players never forget what they have learned.)
Yet other combinations are not interesting, such as the attacker-first sequential game with 
(totally) imperfect information (by the defender), since it coincides with the simultaneous-game case. 
Note that the attacker and defender are not symmetric with respect to hiding/revealing their actions $a$ and $d$, since the knowledge of $a$ affects the game only in the usual sense of game theory (in Phase 1), while the knowledge of $d$ also affects the computation of the payoff (in Phase 2). 
Note that the attacker and defender are not symmetric with respect to hiding/revealing their actions $a$ and $d$, since the knowledge of $a$ affects the game only in the usual sense of game theory, while the knowledge of $d$ also affects the computation of the payoff (cf.  \qm{Phase 2} above). 
Other possible combinations would come from the distinction between behavioral and mixed strategies, but, as we will see, they are always equivalent except in one scenario, so for the sake of conciseness we prefer to treat it as 
a case apart. 

Table~\ref{table:games} lists the meaningful and interesting combinations. 
In Game V we assume imperfect information: the attacker does not 
know the action chosen by the defender. 
In all the other sequential games we assume that the follower has perfect 
information.
In the remaining of this section, we discuss each game individually, using the example of Section~\ref{sec:running-example} as a running example. 

\begin{table}[!tb]
\centering
\renewcommand{\arraystretch}{1.35}
\begin{tabular}{c|c|c|c|c|}
\multicolumn{2}{}{} & \multicolumn{3}{c}{Order of action} \\ \cline{3-5}
\multicolumn{2}{}{} & \multicolumn{1}{|c|}{\textbf{simultaneous}} & \textbf{defender first} & \textbf{attacker first} \\ \cline{2-5}
\multirow{2}{*}{\stackanchor{Defender's}{choice}} & \textbf{visible $\vchoiceop$} & Game I & Game II & Game III \\ \cline{2-5}
&  {\textbf{hidden $\hchoiceop$}} & Game IV & Game V & Game VI \\ \cline{2-5}
\end{tabular}
\renewcommand{\arraystretch}{1}
\caption{Kinds of games we consider. Sequential games have perfect information, except for Game V.}
\label{table:games}
\end{table}

\subsubsection{Game I (simultaneous with visible choice)}

This simultaneous game can be represented by a tuple $(\cald,\, \cala,\, \pay)$. 
As in all games with visible choice $\vchoiceop$, 
the expected payoff $\Pay$ of a mixed strategy profile $(\delta,\alpha)$ 
is defined to be the expected value of $u$, as in traditional game theory: 
\begin{align*}
\Pay(\delta,\alpha) 
\eqdef {\expectDouble{d\leftarrow\delta}{a\leftarrow\alpha}\hspace{-0.5ex}  \pay(d, a)}
=\hspace{-0.5ex} \sum_{\substack{d\in\cald\\ a\in\cala}} \delta(d) \,\alpha(a)\, \pay(d, a),
\end{align*}
where we recall that $\pay(d,a) = \postvf{\pi}{C_{da}}$. 

From Theorem~\ref{theo:convex-V-q}(\ref{enumerate:vg:v:linear}) we derive that
$\Pay(\delta,\alpha) 
= \postvf{\pi}{\VChoiceDouble{d}{\delta}{a}{\alpha} C_{da}} 
$,
and hence the whole system can be equivalently regarded as 
the channel $\VChoiceDouble{d}{\delta}{a}{\alpha} C_{da}$. 
Still from Theorem~\ref{theo:convex-V-q}(\ref{enumerate:vg:v:linear}) we 
can derive that $\Pay(\delta,\alpha)$ is linear in $\delta$ and $\alpha$. 
Therefore the Nash equilibrium can be computed using the standard method 
(cf. Section~\ref{subsec:game-theory}). 
\begin{Example}
Consider the example of Section~\ref{sec:running-example} in the setting of Game I, with uniform prior.
The Nash equilibrium $(\delta^*,\alpha^*)$ can be obtained using the closed formula 
from Section \ref{subsec:game-theory},
and it is given by
$
\delta^*(0) =\alpha^*(0) =\nicefrac{(\nicefrac{2}{3}-1)}{(\nicefrac{1}{2}-1-1+\nicefrac{2}{3})}=\nicefrac{2}{5}.
$
The corresponding payoff is
$
\Pay(\delta^*,\alpha^*)= \nicefrac{2}{5}\,\nicefrac{2}{5}\,\nicefrac{1}{2}+\nicefrac{2}{5}\,\nicefrac{3}{5} + \nicefrac{3}{5}\,\nicefrac{2}{5} + \nicefrac{3}{5}\,\nicefrac{3}{5}\,\nicefrac{2}{3}= \nicefrac{4}{5}
$.
\end{Example}

\subsubsection{Game II (defender-first with visible choice)}

This defender-first sequential game can be represented by a tuple $(\cald,\, \calda,\, \pay)$. 
We will first consider  mixed strategies for the follower (which in this case is the attacker), 
namely strategies of type $\distr(\calda)$. 
Hence a (mixed) strategy profile is of the form $(\delta, \msa)$, with $\delta\in\distr\cald$ and 
$\msa\in\distr(\calda)$, and the corresponding payoff is
\begin{align*}
\Pay(\delta,\msa) 
\eqdef {\expectDouble{d\leftarrow\delta}{\psa\leftarrow\msa} \pay(d, \psa(d))} \allowbreak
= \allowbreak \sum_{\substack{d\in\cald\\ \psa:\calda}} \delta(d)\, \msa(\psa)\, \pay(d, \psa(d)),
\end{align*}
where  $\pay(d, \psa(d)) = \postvf{\pi}{C_{d\psa(d)}}$. \\[1mm]

Again, from Theorem~\ref{theo:convex-V-q}(\ref{enumerate:vg:v:linear}) we derive:
$
\textstyle
\Pay(\delta,\msa) 
= \postvf{\pi}{\VChoiceDouble{d}{\delta}{\psa}{\msa}  C_{d \psa(d)}} 
$
and hence the system can be expressed as a channel 
$\textstyle \VChoiceDouble{d}{\delta}{\psa}{\msa} C_{d \psa(d)}$.
From the same theorem we also derive that $\Pay(\delta,\msa)$ is 
linear in $\delta$ and $\msa$, so the mutually optimal strategies can 
be obtained again by solving the minimax problem.
In this case, however, the solution is particularly simple, because 
there are always deterministic optimal strategy profiles. We first consider the case of attacker's strategies of type 
$\distr(\calda)$.

\begin{Theorem}[Pure-strategy Nash equilibrium in Game II -  strategies of type $\distr(\calda)$]
\label{theo:deterministic-strategies-II}
Consider a defender-first sequential game with visible choice, and attacker's strategies of type 
$\distr(\calda)$. Let $d^*\eqdef \argmin_d \max_a \pay(d,a)$ and let $\psa^*:\calda$ be defined as $\psa^*(d)\eqdef  \argmax_{a} \pay(d,a)$
(if there is more than one  $a$ that maximizes $\pay(d,a)$, we select one of them arbitrarily). 
Then   
for every $\delta\in\distr\cald$ and $\msa\in\distr(\calda)$ we have
$
\Pay(d^*,\msa)\leq \pay(d^*,\psa^*(d^*))\leq \Pay(\delta,\psa^*)
$.
\end{Theorem}
\begin{proof}
Let 
 $\delta$ and  $\msa$ be arbitrary elements of $\distr\cald$ and $\distr(\calda)$, respectively. Then:
\begin{align*}
\Pay(d^*,\msa)
=     & \;\; \sum_{\psa:\calda} \msa(\psa)\, \pay(d^*, \psa(d^*))\\
\leq & \;\; \sum_{\psa:\calda} \msa(\psa)\, \pay(d^*, \psa^*(d^*)) &\mbox{(by definition of $\psa^*$)}\\
= & \;\;  \pay(d^*, \psa^*(d^*))&\mbox{(since $\msa$ is a distribution)}\\
= & \;\; \sum_{d\in\cald} \delta(d)\,\pay(d^*, \psa^*(d^*))  &\mbox{(since $\delta$ is a distribution)}\\
\leq & \;\; \sum_{d\in\cald} \delta(d)\,\pay(d, \psa^*(d))    &\mbox{(by definition of $d^*$)}\\
= & \;\; \Pay(\delta,\psa^*)
\end{align*}
\end{proof}

Hence to find the optimal strategy  it is sufficient for the defender to find  the action $d^*$ which 
minimizes $\max_a \pay(d^*,a)$, while the attacker's optimal choice is the pure strategy $\psa^*$ such that 
$\psa^*(d) = \argmax_a \pay(d,a)$,
where $d$ is the (visible) move by the defender.

\begin{Example}
Consider  the example of Section~\ref{sec:running-example} in the setting of Game II, with uniform prior.
If the defender chooses $0$ then the attacker chooses $1$. 
If the defender chooses $1$ then the attacker chooses $0$.
In both cases, the payoff is $1$. The game has therefore two solutions, $(\delta^*_1,\alpha^*_1)$ and $(\delta^*_2,\alpha^*_2)$, 
with $\delta^*_1(0)=1$, $\alpha^*_1(0)=0$ and $\delta^*_2(0)=0$, $\alpha^*_2(1)=1$.
\end{Example}

Consider now the case of behavioral strategies. 
Following the same line of reasoning as before, 
we can see that under the strategy profile $(\delta,\fsa)$ the system can be expressed as the channel 
$${\displaystyle \VChoice{d}{\delta}\,\VChoice{a}{\fsa(d)} C_{d a}}.$$

Also in this case there are deterministic optimal strategy profiles. 
An optimal strategy for the follower (in this case the attacker) consists in looking at the action $d$ chosen by the leader and then 
selecting with probability $1$ the action $a$ that maximizes $\pay(d,a)$.

\begin{Theorem}[Pure-strategy Nash equilibrium in Game II - strategies of type $\cald \rightarrow \distr(\cala)$]
\label{theo:deterministic-strategies-II-2}
Consider a defender-first sequential game with visible choice, and attacker's strategies of type 
$\cald \rightarrow \distr(\cala)$.
Let $d^*\eqdef \argmin_d \max_a \pay(d,a)$ and let $\phi^*_{\sf a}:\cald \rightarrow \distr(\cala)$ be defined as $\phi^*_{\sf a}(d)(a)\eqdef 1$ if $a= \argmax_{a'} \pay(d,a')$ (if there is more than one such $a$, we select one of them arbitrarily), and $\phi^*_{\sf a}(d)(a)\eqdef 0$ otherwise. 
Then   
for every $\delta\in\distr\cald$ and $\fsa: \cald \rightarrow \distr(\cala)$ we have:
$
\Pay(d^*,\fsa(d^*))\leq \Pay(d^*,\phi^*_{\sf a}(d^*))\leq \Pay(\delta,\phi^*_{\sf a})
$.
\end{Theorem}
\begin{proof}
Let $a^*$ be the action selected by $\phi^*_{\sf a}(d^*)$, i.e.,  $\phi^*_{\sf a}(d^*)(a^*)\eqdef 1$. Then $\pay(d^*,a^*)= \max_{a} \pay(d^*,a)$. Let
$\delta$ and  $\fsa$ be arbitrary elements of $\distr\cald$ and 
$\cald \rightarrow \distr(\cala)$, respectively. 
Then:
\begin{align*}
\Pay(d^*,\fsa(d^*))
=     & \;\; \sum_{a\in\cala} \fsa(d^*)(a)\, \pay(d^*, a)\\
\leq & \;\; \sum_{a\in\cala}  \fsa(d^*)(a)\,\pay(d^*, a^*) &\mbox{(since $\pay(d^*,a^*)= \max_{a} \pay(d^*,a)$)}\\
= & \;\;  \pay(d^*, a^*)&\mbox{(since $\fsa(d^*)$ is a distribution)}\\
= & \;\;  \Pay(d^*,\phi^*_{\sf a}(d^*))&\mbox{(by definition of  $a^*$)}\\
= & \;\; \sum_{d\in\cald} \delta(d)\,\pay(d^*, \phi^*_{\sf a}(d^*))  &\mbox{(since $\delta$ is a distribution)}\\
\leq & \;\; \sum_{d\in\cald} \delta(d)\,\pay(d, \phi^*_{\sf a}(d))    &\mbox{(by definition of $d^*$ and of $\phi^*_{\sf a}$)}\\
= & \;\; \Pay(\delta,\phi^*_{\sf a})
\end{align*}
\end{proof}
As a consequence of Theorems~\ref{theo:deterministic-strategies-II} and \ref{theo:deterministic-strategies-II-2} we can show that in the games we consider the 
 payoff of the optimal mixed and behavioral strategy profiles coincide. 
 Note that this result 
could also be derived from the result from standard  game theory which states that, in the cases we
consider, for any behavioral strategy there is a mixed strategy that yields the same payoff, and viceversa~\cite{Osborne:94:BOOK}. However, the proof of ~\cite{Osborne:94:BOOK} relies on Khun's theorem, which is non-constructive (and rather complicated, because it is for more general cases). In our scenario the proof is very simple, as we will see in the following corollary. 
Furthermore, since such a result does not hold for leakage games with hidden choice, we think it will be useful to show the proof formally in order to analyze the difference. 
 
 \begin{Corollary}[Equivalence of optimal strategies of types $\distr(\calda)$ and $\cald \rightarrow \distr(\cala)$ in Game II]
\label{coro:equivalence-II}
Consider a defender-first sequential game with visible choice, and let $d^*$, $\psa^*$ and $\phi^*_{\sf a}$ be defined as in Theorems~\ref{theo:deterministic-strategies-II} and \ref{theo:deterministic-strategies-II-2} respectively. 
Then   
$
\pay(d^*,\psa^*(d^*))= \Pay(d^*,\phi^*_{\sf a}(d^*)).
$
\end{Corollary}
\begin{proof}
The result follows immediately by observing that $\pay(d^*,\psa^*(d^*))=\max_a \pay(d^*,a)= \pay(d^*,a^*)= \Pay(d^*,\phi^*_{\sf a}(d^*))$.
\end{proof}

\subsubsection{Game III (attacker-first with visible choice)}

This game is also a sequential game, but with the attacker as the leader. 
Therefore it can be represented as tuple of the form $(\calad,\,  \cala,\, \pay)$. 
It is the same as Game II, except that the roles of the attacker and the 
defender are inverted. 
In particular, the payoff of a mixed strategy profile 
$(\msd, \alpha)\in \distr(\calad)\times \distr\cala$ is given by
\begin{align*}
\Pay(\msd,\alpha)  
\eqdef {\expectDouble{\psd\leftarrow\msd}{a\leftarrow\alpha}\hspace{-0.5ex}
  \pay(\psd(a),a)} = {\sum_{\substack{\psd:\calad \\ a\in\cala}}} \msd(\psd) 
  \, \alpha(a)\, \pay(\psd(a), a) \allowbreak 
\end{align*}
and by Theorem~\ref{theo:convex-V-q}(\ref{enumerate:vg:v:linear})  the whole system can be equivalently regarded as channel 
${\displaystyle \VChoiceDouble{\psd}{\msd}{a}{\alpha}  C_{\psd(a) a}}$. 
For a behavioral strategy $(\fsd, \alpha)\in (\cala\rightarrow \distr(\cald))\times \distr\cala$, the payoff  is given by
\begin{align*}
\Pay(\fsd,\alpha)  
\eqdef {\expect_{a\leftarrow\alpha}\;\;\expect_{d\leftarrow\fsd(a)}\hspace{-0.5ex}
  \pay(d,a)} = {\sum_{a\in\cala}\alpha(a)}{\sum_{d\in\cald}\fsd(a)(d) }\pay(d, a) \allowbreak 
\end{align*}
and by Theorem~\ref{theo:convex-V-q}(\ref{enumerate:vg:v:linear})  the whole system can be equivalently regarded as channel 
${\displaystyle \VChoice{a}{\alpha} \;\; \VChoice{d}{\fsd(a)}  C_{d a}}$.

Obviously, the same results that we have obtained in previous section for Game II hold also for Game III, with the role of attacker and defender switched. We collect all these results in the following theorem. 

\begin{Theorem}[Pure-strategy Nash equilibria in Game  III and equivalence of $ \distr(\calad)$ and  $(\cala\rightarrow \distr(\cald))$]
\label{theo:deterministic-strategies-III}
Consider a defender-first sequential game with visible choice. 
Let $a^*\eqdef \argmax_a \min_d \pay(d,a)$. Let $\psd^*:\calad$ be defined as $\psd^*(a)\eqdef  \argmin_{d} \pay(d,a)$, and 
let $\phi^*_{\sf d}:\cala \rightarrow \distr(\cald)$ be defined as $\phi^*_{\sf d}(a)(d)\eqdef 1$ if $d= \argmin_{d'} \pay(d',a)$. 
Then:
\begin{enumerate}
\item For every $\alpha\in\distr\cala$ and $\msd\in\distr(\calad)$ we have
$
\Pay(\psd^*,\alpha)\leq \pay(\psd^*(a^*),a^*)\leq \Pay(\msd,a^*)
$.\\[-1ex]
\item For every $\alpha\in\distr\cala$ and $\fsd: \cala \rightarrow \distr(\cald)$ we have:
$
\Pay(\phi^*_{\sf d},\alpha)\leq \Pay(\phi^*_{\sf d}(a^*),a^*)\leq \Pay(\fsd(a^*),a^*)
$.\\[-1ex]
\item $\pay(\psd^*(a^*),a^*)= \Pay(\phi^*_{\sf d}(a^*),a^*)$.
\end{enumerate}
\end{Theorem}
 \begin{proof}
These results can be proved by following the same lines as the proofs of Theorems~\ref{theo:deterministic-strategies-II} and \ref{theo:deterministic-strategies-II-2}, and Corollary~\ref{coro:equivalence-II}.
 \end{proof}
 
\begin{Example}
Consider now the example of Section~\ref{sec:running-example} in the setting of Game III, with uniform prior.
If the attacker chooses $0$ then the defender chooses $0$ and the payoff is $\nicefrac{1}{2}$.  
If the attacker chooses $1$ then the defender  chooses $1$ and the payoff is $\nicefrac{2}{3}$.
The latter case is more convenient for the attacker, hence the solution of the game is the strategy profile $(\delta^*,\alpha^*)$ 
with $\delta^*(0)=0$, $\alpha^*(0)=0$.
\end{Example}

\subsubsection{Game IV (simultaneous with hidden choice)}

The simultaneous game with hidden choice is 
a tuple $(\cald,\cala,\pay)$.
However, \emph{it is not an ordinary game} in the sense that 
\emph{the payoff a mixed strategy profile cannot be defined by averaging
the payoff of the corresponding pure strategies}. 
More precisely, the payoff of a mixed profile is defined by
averaging on the strategy of the attacker, but not on that 
of the defender. 
In fact, when hidden choice is used, there is an additional level 
of uncertainty in the relation between the observables and the secret 
from the point of view of the attacker, since he is not sure about which 
channel is producing those observables. 
A mixed strategy $\delta$ for the defender produces a convex combination 
of  channels  (the channels associated to the pure strategies) with the 
same coefficients, and we know from previous sections that  the vulnerability 
is a convex function of the channel, and in general is not linear. 

In order to define the payoff of a mixed strategy  profile  $(\delta,\alpha)$, 
we need therefore to consider the channel that the attacker perceives given his 
limited knowledge. Let us assume that the action that the attacker draws from 
$\alpha$  is $a$. 
He does not know the action of the defender, but we can assume 
that he knows his strategy (each player can derive the optimal strategy of the 
opponent, under the assumption of common knowledge and rational players).

The channel the attacker will see is 
$\HChoice{d}{\delta} C_{d a}$, obtaining a corresponding payoff of 
$\postvf{\pi}{{\HChoice{d}{\delta}} {C_{d a}}}$.
By averaging on the strategy of the attacker we obtain 
\begin{align*}
\Pay(\delta,\alpha) 
\eqdef  
{\expectDouble{a\leftarrow\alpha}{}\hspace{-1ex} \; \postvf{\pi}{{\HChoice{d}{\delta}} {C_{d a}}}}
= \sum_{a\in\cala} \,\alpha(a)\, \postvf{\pi}{{\HChoice{d}{\delta}} {C_{d a}}}
{.}
\end{align*}
From Theorem~\ref{theo:convex-V-q}(\ref{enumerate:vg:v:linear}) we derive:
$
\Pay(\delta,\alpha) 
= \postvf{\pi}{\VChoice{a}{\alpha}\, {\HChoice{d}{\delta}} {C_{d a}}} 
$
and hence the whole system can be equivalently regarded as channel 
${\VChoice{a}{\alpha}\, {\HChoice{d}{\delta}} {C_{d a}}}$.
Note that, by Proposition~\ref{prop:reorganization-operators}(\ref{item:c}), 
the order of the operators is interchangeable, and the system can be 
equivalently regarded as $ {\HChoice{d}{\delta}}\,{\VChoice{a}{\alpha} {C_{d a}}}$.
This shows the robustness of this model. 

From Corollary~\ref{cor:concave-convex-V} we  derive that $\Pay(\delta,\alpha)$ is convex in $\delta$ and linear in $\eta$, hence we can compute the Nash equilibrium by the minimax method.

\begin{Example}\label{exa:exempio}
Consider now the example of Section~\ref{sec:running-example} in the setting of Game IV.
For $\delta\in\distr\cald$ and $\alpha\in\distr\cala$, let $p=\delta(0)$ and $q=\alpha(0)$. 
The system can be represented by the channel 
$(C_{00}\hchoice{p} C_{10})\vchoice{q}(C_{01}\hchoice{p} C_{11})$
represented below.
\begin{align*}
\begin{array}{|c|c|c|}
\hline
C_{00}  \hchoice{p} C_{10}& y=0 & y=1 \\ \hline
x=0    & p & \overline{p} \\
x=1    & 1 & 0 \\ \hline
\end{array}
\quad
\vchoice{q}
\quad
\begin{array}{|c|c|c|}
\hline
C_{01}\hchoice{p} C_{11}& y=0 & y=1 \\ \hline
x=0    & \nicefrac{1}{3}+ \nicefrac{2}{3}\;p &   \nicefrac{2}{3}- \nicefrac{2}{3}\;p \\
x=1    & \nicefrac{2}{3}- \nicefrac{2}{3}\;p  & \nicefrac{1}{3}+ \nicefrac{2}{3}\,p \\ \hline
\end{array}
\end{align*}
\noindent For uniform $\pi$, we have
$\postvf{\pi}{ C_{00} \hchoice{p} C_{10}}{=}1-\nicefrac{1}{2}\,p$, 
while
$\postvf{\pi}{ C_{10}  \hchoice{p} C_{11}}$ is equal to
$
\nicefrac{2}{3}- \nicefrac{2}{3}\,p$ if $p\leq \nicefrac{1}{4}$,
and equal to 
$\nicefrac{1}{3}+ \nicefrac{2}{3}\,p$ 
if $p> \nicefrac{1}{4}$.
Hence the payoff, expressed in terms of $p$ and $q$, is
$\Pay(p,q) = q(1-\nicefrac{1}{2}\,p) + \overline{q} (\nicefrac{2}{3} - \nicefrac{2}{3}\,p)$
if $p\leq \nicefrac{1}{4}$, and 
$\Pay(p,q) = q(1-\nicefrac{1}{2}\,p) + \overline{q} (\nicefrac{1}{3} + \nicefrac{2}{3}\,p)$ 
if $p> \nicefrac{1}{4}$.
The Nash equilibrium can be computed by imposing that the partial derivatives of 
$\Pay(p,q)$ with respect to $p$ and $q$ are both $0$, which means that we are in a  saddle point.  
We have:
\[
\renewcommand{\arraystretch}{1.7}
\frac{\partial \Pay(p,q)}{\partial q}=\left\{\begin{array}{ll}
		\frac{1}{3}+\frac{1}{6}\,p&, \text{ if } p\leq\frac{1}{4}\\
		\frac{2}{3}-\frac{7}{6}\,p&, \text{ if } p>\frac{1}{4}\\
		\end{array}
		\right.
\qquad
\frac{\partial \Pay(p,q)}{\partial p}=\left\{\begin{array}{ll}
		-\frac{2}{3}+\frac{1}{6}\,q&, \text{ if } p\leq\frac{1}{4}\\
		\frac{2}{3}-\frac{7}{6}\,q&, \text{ if }  p>\frac{1}{4}\\
		\end{array}
		\right.		
\renewcommand{\arraystretch}{1}
\]		
We can see that the equations $\nicefrac{\partial \Pay(p,q)}{\partial q}=0$ and $\nicefrac{\partial \Pay(p,q)}{\partial p}=0$ do not have solutions in $[0,1]$ for $p\leq \nicefrac{1}{4}$, 
while for $p> \nicefrac{1}{4}$ they have solution $p^*=q^*=\nicefrac{4}{7}$.
The pair $(p^*, q^*)$ thus constitutes the Nash equilibrium, and the corresponding payoff is $\Pay(p^*,q^*) = \nicefrac{5}{7}$. 
\end{Example}

\subsubsection{Game V (defender-first with hidden choice)}

This is a defender-first sequential game with imperfect information, hence it can be represented as a tuple of the form  
$(\cald,\, \infoseta\rightarrow\cala,\, \pay)$,
where $\infoseta$ is a partition of $\cald$. 
Since we are assuming perfect recall, and the attacker does not know anything about the 
action chosen by the defender in Phase 2, i.e., at the moment of the attack (except the probability distribution determined by his strategy), we must assume that the attacker does not know anything in Phase 1 either. Hence the indistinguishability relation must be total, i.e.,  $\infoseta=\{\cald\}$. But $\{\cald\}\rightarrow\cala$ is equivalent to $\cala$, hence this kind of game is equivalent to Game IV.
It is also a well known fact in game theory that when in a sequential game 
the follower does not know the leader's move before making his choice, 
the game is equivalent to a simultaneous game.\footnote{However, one could
argue that, since the defender has already committed, the attacker does not
need to perform the action corresponding to the Nash equilibrium, any 
payoff-maximizing solution would be equally good for him.
}

\subsubsection{Game VI (attacker-first with hidden choice)}

This game is also a sequential game with the attacker as the leader, hence it is a tuple of the form
$(\calad,\,  \cala,\, \pay)$. 
It is  similar to Game III, except that the payoff is convex on the strategy of the defender, instead of linear. 
We will see, however,  that this causes quite some deviation from the properties of Game III, and from standard game theory.

The payoff of the mixed strategy profile  $(\msd, \alpha)\in \distr(\calad)\times \distr\cala$  is
\begin{align*}
\Pay(\msd,\alpha)  
\eqdef {\expectDouble{a\leftarrow\alpha}{}\hspace{-1ex} \; \postvf{\pi}{{\HChoice{\psd}{\msd}} {C_{\psd(a) a}}}} 
= {\postvf{\pi}{ \VChoice{a}{\alpha}\;\HChoice{\psd}{\msd}  C_{\psd(a) a} } },
\end{align*}
so the whole system can be equivalently regarded as channel 
$\displaystyle \VChoice{a}{\alpha}\;\HChoice{\psd}{\msd}  C_{\psd(a) a}$.

The first important difference from Game III is that in Game VI   there may not exist optimal strategies, either mixed or behavioral, that are deterministic for the defender.
On the other hand, for the attacker there are always deterministic optimal strategies, and this is true independently from whether the defender uses mixed or behavioral strategies.  

To show the existence of deterministic optimal strategies for the attacker, let us first introduce some standard notation for functions: 
given a variable $x$ and an expression $M$, $\lambda x.M$ represents the function that 
on the argument $x$ gives as result the value of $M$. Given  two sets   $X$ and $Y$ where $Y$ is provided with an ordering $\leq$, the 
\emph{point-wise ordering} on $X\rightarrow Y$ is defined as follows: for  $f,g: X\rightarrow Y$, $f\leq g$ if and only if 
$\forall x \in X. \, f(x)\leq g(x)$.

\begin{Theorem}[Attacker's pure-strategy Nash equilibrium in Game VI]\label{Theo:deterministic solution Game VI}
Consider an attacker-first sequential game with hidden choice.
\begin{description}
\item[1. Mixed strategies -- type $\distr(\calad)$.] Let  
$a^*\eqdef \argmax_a\min_{\msd}\postvf{\pi}{{\HChoice{\psd}{\msd}} {C_{\psd(a) a}}}$, and let 
$\sigma^*_{\sf d}\eqdef  \argmin_{\msd} \lambda a. \postvf{\pi}{{\HChoice{\psd}{\msd}} {C_{\psd(a) a}}}$.
Then, for every $\alpha\in\distr\cala$ and $\msd\in\distr(\calad)$ we have
\[
\Pay(\sigma^*_{\sf d},\alpha)\leq \Pay(\sigma^*_{\sf d},a^*)\leq \Pay(\msd,a^*)
\]
\item[2. Behavioral strategies -- type $\cala\rightarrow \distr(\cald)$.] Let  
$a^*\eqdef \argmax_a\min_{\delta}\postvf{\pi}{{\HChoice{d}{\delta}} {C_{d a}}}$, and let
$\phi^*_{\sf d}\eqdef  \argmin_{\fsd} \lambda a. \postvf{\pi}{{\HChoice{d}{\fsd(a)}} {C_{d a}}}$. (The minimization is with respect to the point-wise ordering.) 
Then, for every $\alpha\in\distr\cala$ and $\fsd:\cala\rightarrow \distr(\cald)$ we have
\[
\Pay(\phi^*_{\sf d},\alpha)\leq \Pay(\phi^*_{\sf d},a^*)\leq \Pay(\fsd,a^*)
\]
\end{description}
\end{Theorem}
\begin{proof}
\begin{enumerate}
\item 
Let $\alpha$ and  $\msd$ be arbitrary elements of $\distr\cala$ and $\distr(\calad)$, respectively. Then:
\begin{align*}
\Pay(\sigma^*_{\sf d},\alpha)
=& \;\; \sum_{a\in\cala} \alpha(a)\, \postvf{\pi}{{\HChoice{\psd}{\sigma^*_{\sf d}}} {C_{\psd(a) a}}}\\
\leq & \;\; \sum_{a\in\cala} \alpha(a)\, \postvf{\pi}{{\HChoice{\psd}{\sigma^*_{\sf d}}} {C_{\psd(a^*) a^*}}}&\mbox{(by definition of $a^*$ and $\sigma^*_{\sf d}$)}\\
=& \;\;  \postvf{\pi}{{\HChoice{\psd}{\sigma^*_{\sf d}}} {C_{\psd(a^*) a^*}}} \;(\; = \;  \Pay(\sigma^*_{\sf d},a^*)) &\mbox{(since  $\alpha$ is a distribution)}\\
\leq & \;\; \postvf{\pi}{{\HChoice{\psd}{\msd}} {C_{\psd(a^*) a^*}}}&\mbox{(by definition of  $\sigma^*_{\sf d}$)}\\
=& \;\; \Pay(\msd,a^*)
\end{align*}
\item
Let $\alpha$ and  $\fsd$ be arbitrary elements of $\distr\cala$ and $\cala\rightarrow \distr(\cald)$, respectively. Then:
\begin{align*}
\Pay(\phi^*_{\sf d},\alpha)
=& \;\; \sum_{a\in\cala} \alpha(a)\, \postvf{\pi}{{\HChoice{d}{\phi^*_{\sf d}(a)}} {C_{d a}}}\\
\leq & \;\; \sum_{a\in\cala} \alpha(a)\, \postvf{\pi}{{\HChoice{d}{\phi^*_{\sf d}(a^*)}} {C_{d a^*}}}&\mbox{(by definition of $a^*$ and $\phi^*_{\sf d}$)}\\
=& \;\;   \postvf{\pi}{{\HChoice{d}{\phi^*_{\sf d}(a^*)}} {C_{d a^*}}} \;(\; = \;  \Pay(\phi^*_{\sf d},a^*)) &\mbox{(since  $\alpha$ is a distribution)}\\
\leq & \;\; \postvf{\pi}{{\HChoice{d}{\fsd(a^*)}} {C_{d a^*}}}&\mbox{(by definition of  $\phi^*_{\sf d}$)}\\
=& \;\; \Pay(\fsd,a^*)
\end{align*}
\end{enumerate}
~
\end{proof}

We show now, with the following example, that the optimal strategies for the defender are necessarily probabilistic.  
\begin{Example}
Consider the channel matrices $C_{ij}$ defined in Section~\ref{sec:running-example},  and define the 
following new channels: $C'_{00} = C'_{11} = C_{01}$ and  $C'_{10} = C'_{01}= C_{10}$. 
Define $D_p$ as the result of the hidden choice, with probability $p$,  between $C'_{00}$ and  $C'_{10}$, i.e., $D_p \eqdef C'_{00}\hchoice{p}C'_{10}$, 
and observe that $D_p[0,0] = D_p[1,1]  = p$, and $D_p[1,0] = D_p[0,1]  = 1-p$.
Furthermore, $D_p= C'_{01}\hchoice{1-p}C'_{11}$. 
The vulnerability of $D_p$, for uniform $\pi$, is $\postvf{\pi}{D_p} = 1-\nicefrac{1}{2}p$ for $p\leq\nicefrac{1}{2}$ and $\postvf{\pi}{D_p} = p$ for $p>\nicefrac{1}{2}$, hence, 
 independently from the choice of the attacker,  the best strategy for the defender is to choose $p=\nicefrac{1}{2}$. 
 Every other value for $p$ gives a strictly higher vulnerability.   
Therefore, the best mixed strategy for the defender is $\sigma^*_{\sf d}$ defined as $\sigma^*_{\sf d}(\lambda_a.0) = \sigma^*_{\sf d}(\lambda_a.1)= \nicefrac{1}{2}$. 
Similarly, the best behavioral strategy for the defender is  $\phi^*_{\sf d}$ defined as $\phi^*_{\sf d}(0) = \phi^*_{\sf d}(1)= \lambda_d. \nicefrac{1}{2}$. 
\end{Example}

The second important difference from Game III is that in Game VI 
behavioral strategies and mixed strategies  are not necessarily equivalent. 
More precisely, there are cases in which the optimal strategy profile yields a different payoff depending on whether 
the defender adopts mixed strategies or behavioral ones. 
The following  is an example in which this difference manifests itself. 

\begin{Example}\label{eg:differ:mixed:behavioral}
Consider again the example of Section~\ref{sec:running-example}, this time in the setting of Game VI, and still with uniform prior $\pi$.
Let us analyze first the case in which the defender uses behavioral strategies. 
\begin{description}
\item[Behavioral strategies -- type $\cala\rightarrow \distr(\cald)$.] 
If the attacker chooses $0$, which corresponds to committing to the system $C_{00}\hchoice{p} C_{10}$, 
then the defender will choose $p=\nicefrac{1}{4}$, which minimizes its vulnerability. 
If he chooses $1$, which corresponds to committing to the system $C_{01}\hchoice{p} C_{11}$, then
the defender will choose $p=1$, which minimizes the vulnerability. 
In both cases, the leakage is $p=\nicefrac{1}{2}$, hence both these strategies are solutions to the minimax. 
Note that in the first case the strategy of the defender is probabilistic, while that of the attacker is  pure in both cases. 
\\[-1ex]
\item[Mixed strategies -- type $\distr(\calad)$.] 
Observe that there are only four possible pure strategies  for the defender, 
corresponding to the four functions $f_{ij}:\calad$ for $i,j\in\{0,1\}$ defined as $f_{ij}(a)\eqdef i$ if $i=j$ and $f_{ij}(a)\eqdef a\oplus i$ if $i\neq j$. 
Consider a distribution $\msd\in \distr(\calad)$ and let $p_{ij}\eqdef  \msd(f_{ij})$. Then we have $p_{ij} \geq 0$ and $\sum_{i,j} p_{ij} = 1$. 
Observe that the  attacker's  choice $a=0$ determines the matrix $C_{00}\hchoice{p}C_{10}$, with $p= p_{00} + p_{10}$, whose vulnerability is 
 $\postvf{\pi}{C_{00}\hchoice{p}C_{10}}= 1-\nicefrac{1}{2}p$. 
 On the other hand, the attacker's choice $a=1$ determines the matrix $C_{01}\hchoice{p'}C_{11}$, with $p'= p_{00} + p_{01}$,  whose vulnerability is 
 $\postvf{\pi}{C_{01}\hchoice{p'}C_{11}}= \nicefrac{2}{3}-\nicefrac{2}{3}p$ for $p'\leq \nicefrac{1}{4}$, and $\postvf{\pi}{C_{01}\hchoice{p'}C_{11}}= \nicefrac{1}{3}+\nicefrac{2}{3}p$ for $p'> \nicefrac{1}{4}$. 
 By geometrical considerations (cf. the red dashed line in Figure~\ref{fig:graphs}) we can see that the optimal solutions for the defender are all those strategies which 
 give $p=\nicefrac{6}{7}$ and $p'= \nicefrac{1}{7}$,  which yield payoff $\nicefrac{4}{7}$.
 
\end{description}
\end{Example}

\begin{figure}[tb]
\centering
\begin{minipage}[c]{0.7\linewidth}
\quad\includegraphics[width=\linewidth]{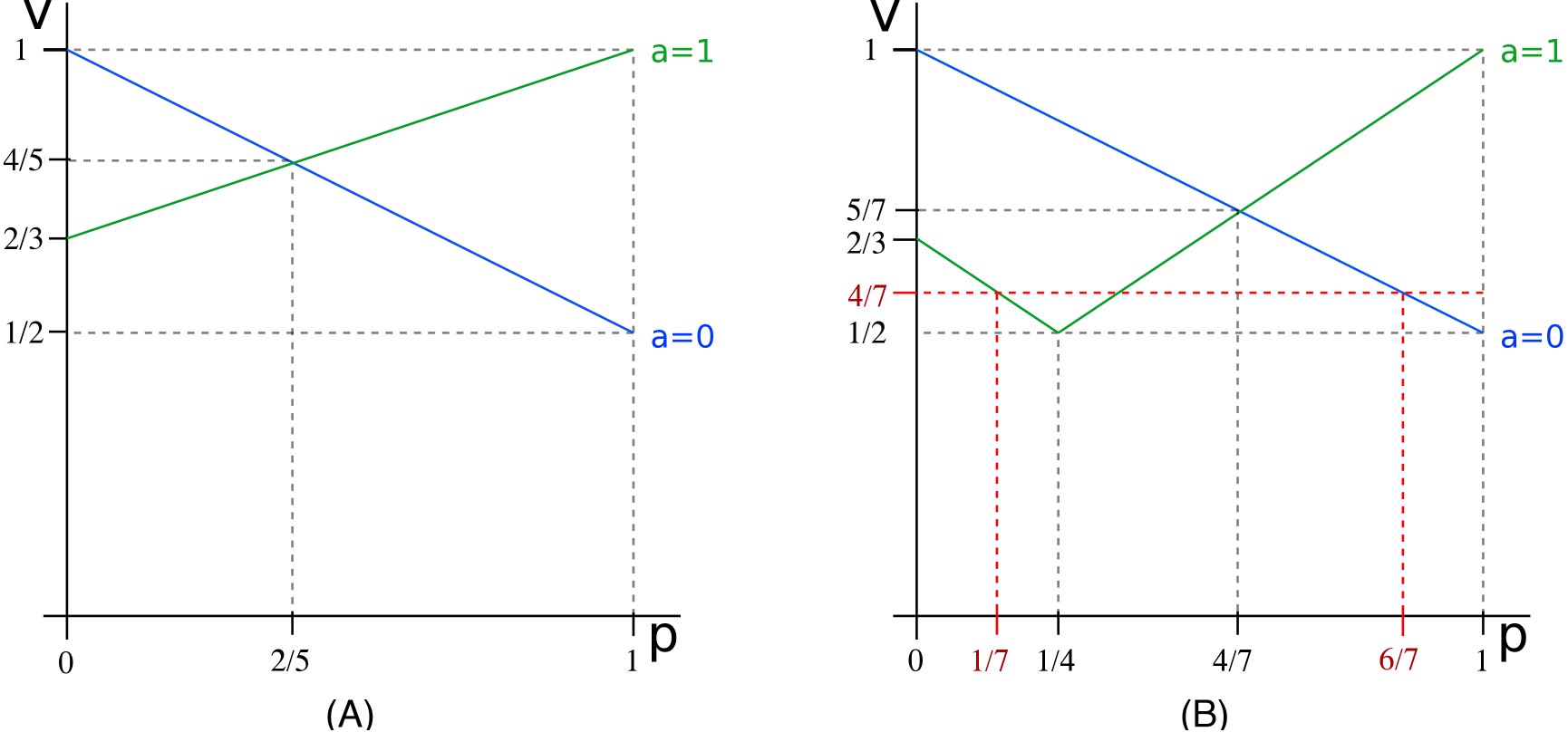}
\end{minipage}
\hfill
\begin{minipage}[c]{0.2\linewidth}
\begin{tabular}[h]{|c|c|}
\hline
Game&$\Pay$\\
\hline
I& $\nicefrac{4}{5}$\\
\hline
II&$1$\\
\hline
III& $\nicefrac{2}{3}$\\
\hline
IV& $\nicefrac{5}{7}$\\
\hline
V& $\nicefrac{5}{7}$\\
\hline
VI$_{\sf m}$& $\nicefrac{4}{7}$\\
\hline
VI$_{\sf b}$& $\nicefrac{1}{2}$\\
\hline
\end{tabular}
\end{minipage}
\caption{Summary of the results for the running example introduced in Section~\ref{sec:running-example}, for uniform prior. Graph (A) is for the case of visible choice: it represents the Bayes vulnerability $\vf{}{}$ of $C_{00}\vchoice{p}C_{10}$ and of $C_{01}\vchoice{p}C_{11}$ (cases $a=0$ and $a=1$ respectively), as a function of $p$. Graph (B) is for the case of hidden choice and it represents the vulnerability of $C_{00}\hchoice{p}C_{10}$ and of $C_{01}\hchoice{p}C_{11}$ as a function of $p$.  
The table on the right gives the payoff in correspondence of the Nash equilibrium for the various games. VI$_{\sf m}$ and VI$_{\sf b}$ represent the attacker-first sequential games with defender strategy of type $\distr(\calad)$ (mixed) and $\cala\rightarrow \distr(\cald)$ (behavioral), respectively.}
\label{fig:graphs}
\end{figure}

The facts that behavioral and mixed strategies are not equivalent is related to the non-existence of deterministic optimal strategies. 
In fact, it is easy to see that  from a behavioral deterministic strategy we can construct a (deterministic) mixed strategy, and vice versa. 

Figure~\ref{fig:graphs} illustrates the graphs of the vulnerability of the various channel compositions, and summarizes the results of this section.

\section{Comparing the leakage games}
\label{sec:comparing-games}
In previous section we have computed the vulnerability  for the running example in the various kind of games introduced in Section~\ref{sec:games-setup}. The values we have  obtained, listed in 
decreasing order, are as follows:
${\rm II:} \, 1;\, \allowbreak 
{\rm I:} \, \nicefrac{4}{5};\, \allowbreak 
{\rm IV:} \, \nicefrac{5}{7};\, \allowbreak 
{\rm V:} \, \nicefrac{5}{7};\, \allowbreak 
{\rm III:} \, \nicefrac{2}{3};\, \allowbreak 
{\rm VI}_{\sf m} {\rm :} \, \nicefrac{4}{7};\, \allowbreak 
{\rm VI}_{\sf b} {\rm :} \, \nicefrac{1}{2} \allowbreak 
$.
This order is not accidental: in this section we will prove that  some of these relations between  games hold
for any vulnerability function, and for any prior. 
These results will allow us to reason about which 
kind of scenarios and compositions are more convenient for the defender, or, vice versa, for the attacker.

\subsection{Simultaneous games vs. sequential games}
The relations between II, I, and III, and between IV-V and VI$_{\sf m}$ are typical in game theory: in any zero-sum sequential game the leader's payoff is
less than or equal to his payoff in the corresponding simultaneous game. In fact, by acting first, the leader commits to an action, and this commitment can be exploited by the attacker to choose the best possible strategy relatively to that action\footnote{The fact that the leader has a disadvantage may seem counterintuitive because in many real games it is the opposite, the player who moves first has an advantage. Such discrepancy is  due to the fact that these games feature preemptive moves, i.e.  moves that, when 
made by one player, make impossible other moves for the other player. The games we are considering in this paper, on the contrary, do not consider preemptive moves.}.
In the following propositions we give the precise formulation of these results in our framework, and we show how they can be derived formally.
\begin{restatable}[Game II $\geqslant$ Game I]{Proposition}{order}
\label{theo:order:II-I}
\begin{align*}
\min_\delta\max_\msa \postvf{\pi}{\VChoiceDouble{d}{\delta}{\psa}{\msa}  C_{d \psa(d)}} 
&\geq
\min_\delta\max_\alpha \postvf{\pi}{\VChoiceDouble{d}{\delta}{a}{\alpha} C_{da}}
\end{align*}
\end{restatable}
\begin{proof}
We prove the first inequality as follows. 
Independently of $\delta$, consider the attacker's strategy 
$\sigma^*_{\sf a}$ 
that assigns probability $1$ to the function $\psa^*$ defined as
$
\psa^*(d) = \argmax_a \postvf{\pi}{C_{d a}} 
$
for any $d\in\cald$.
Then we have that:
\begin{align*}
\min_\delta\max_\msa \postvf{\pi}{\VChoiceDouble{d}{\delta}{\psa}{\msa}  C_{d \psa(d)}} 
\geq&\,\, 
\min_\delta \postvf{\pi}{\VChoiceDouble{d}{\delta}{\psa}{{\sigma^*_{\sf a}}} \, C_{d \psa(d)}} 
& \text{(by maximization on $\msa$)}
\\ =&\,\, \min_\delta \postvf{\pi} {\VChoice{d}{\delta} \, C_{d \psa^*(d)}}
& \text{(by definition of $\sigma^*_{\sf a}$)}
\\ =&\,\, 
\min_\delta \sum_{d} \delta(d) \, \postvf{\pi}{C_{d \psa^*(d)}}
& \text{(by Theorem~\ref{theo:convex-V-q}(\ref{enumerate:vg:v:linear}))}
\\ =&\,\, 
\min_\delta \sum_{d} \delta(d) \max_\alpha \sum_{a} \alpha(a)\, \postvf{\pi}{C_{d \psa^*(d)}}
& \text{(since  $\alpha$ is a distribution)}
\\ \geq&\,\, 
\min_\delta \sum_{d} \delta(d) \max_\alpha \sum_{a} \alpha(a) \,\postvf{\pi}{C_{da}}
& \text{(by definition of $\psa^*$)}
\\ \geq&\,\, 
\min_\delta \max_\alpha \sum_{d} \delta(d) \sum_{a} \alpha(a)\, \postvf{\pi}{C_{da}}
\\ =&\,\, 
\min_\delta\max_\alpha \postvf{\pi}{\VChoiceDouble{d}{\delta}{a}{\alpha} C_{da}} 
& \text{(by Theorem~\ref{theo:convex-V-q}(\ref{enumerate:vg:v:linear}))}
\end{align*}
\end{proof}
Note that the strategy $\sigma^*_{\sf a}$ is optimal  for the attacker, so the first of the above inequalities is actually an equality. 
It is easy to see that the second inequalities is an equality as well, because of the maximization on $\alpha$. 
Therefore, the only  inequalities that may be strict is the third one, and the reason why it may be strict is that 
in the left-hand side $\alpha$ depends on $d$ (and  on $\delta$), while on the right-hand side $\alpha$ 
depends on $\delta$, but not  the actual $d$ (that will be sampled from $\delta$). This corresponds to the fact that 
in the defender-first sequential game the attacker chooses his strategy after he knows the action $d$ chosen by the defender, 
while in the simultaneous game the attacker knows the strategy of the defender (i.e., the distribution $\delta$ he will use to choose probabilistically his actions), but not the actual action $d$ that the defender will choose. 

Analogous considerations can be done for the simultaneous versus the attacker-first case, that we will examine next. 

\begin{restatable}[Game I $\geqslant$ Game III]{Proposition}{order}
\label{theo:order:I-III}
\begin{align*}
\min_\delta\max_\alpha \postvf{\pi}{\VChoiceDouble{d}{\delta}{a}{\alpha} C_{da}}
\,\,\geq\,\,
\max_\alpha  \min_{\msd}{\postvf{\pi}{\VChoiceDouble{\psd}{\msd}{a}{\alpha}  C_{\psd(a) a}} }
\end{align*}
\end{restatable}

\begin{proof}
Independently of $\alpha$, consider the defender's strategy $\sigma^*_{\sf d}$ that assigns probability $1$ to the function $\psd^*$ defined as
$
\psd^*(a) = \argmin_d \postvf{\pi}{C_{d a}} 
$
for any $a\in\cala$.
Then we have that:
\begin{align*}
\min_\delta\max_\alpha \postvf{\pi}{\VChoiceDouble{d}{\delta}{a}{\alpha} C_{da}}
=&\,\, 
\min_\delta\max_\alpha \sum_{d} \delta(d) \sum_{a} \alpha(a) \,\postvf{\pi}{C_{da}}
& \text{(by Theorem~\ref{theo:convex-V-q}(\ref{enumerate:vg:v:linear}))}
\\ =&\,\, 
\max_\alpha\min_\delta \sum_{d} \delta(d) \sum_{a} \alpha(a) \,\postvf{\pi}{C_{da}}
& \text{(by Theorem~\ref{theo:vonneumann})}
\\ \geq&\,\, 
\max_\alpha \sum_{a} \alpha(a) \min_{d}\, \postvf{\pi}{C_{da}}
\\ =&\,\, 
\max_\alpha \sum_{a} \alpha(a) \, \postvf{\pi}{C_{\psd^*(a) a}}
& \text{(by definition of $\psd^*$)}
\\ =&\,\, 
\max_\alpha \sum_{a}\alpha(a) \sum_{{\psd}} \sigma^*_{\sf d}(\psd)\, \postvf{\pi}{C_{\psd(a) a}}
& \text{(by definition of $\sigma^*_{\sf d}$)}
\\ =&\,\, 
\max_\alpha \postvf{\pi}{\VChoiceDouble{\psd}{{\sigma^*_{\sf d}}}{a}{\alpha}  C_{\psd(a) a}}
& \text{(by Theorem~\ref{theo:convex-V-q}(\ref{enumerate:vg:v:linear}))}
\\ \geq&\,\, 
\max_\alpha  \min_{\msd}{\postvf{\pi}{\VChoiceDouble{\psd}{\msd}{a}{\alpha}  C_{\psd(a) a}} }
& \text{(by minimization on $\msd$)}
\end{align*}
\end{proof}
Again,  the strategy $\sigma^*_{\sf d}$ is optimal  for the attacker, so the last of the above inequalities is actually an equality. 
Therefore, the only  inequality  that may be strict is the first one 
and the strictness is due to the fact that in the left-hand side $\delta$ does not depend on  $a$  while in the right-hand side it does. 
Intuitively, this corresponds to the intuition that if the defender knows the action of the attacker then it may be able to choose a better strategy to reduce the leakage.

\begin{restatable}[Game IV $\geqslant$ Game VI$_{\sf m}$]{Proposition}{order}
\label{theo:order:IV-VI}
\begin{align*}
\min_\delta\max_\alpha\postvf{\pi}{\VChoice{a}{\alpha}\, {\HChoice{d}{\delta}} {C_{d a}}} &\geq
\max_\alpha  \min_{\msd} {\postvf{\pi}{ \VChoice{a}{\alpha}\;\HChoice{\psd}{\msd}  C_{\psd(a) a} } }
\end{align*}
\end{restatable}
\begin{proof}
Given $\alpha\in \distr\cala$, let $\delta^*_\alpha\eqdef \min_\delta\sum_{a} \alpha(a) \postvf{\pi}{\HChoice{d}{\delta}C_{da}}$. 
For any $d\in\cald$, let $f_d$ be the constant function defined as $f_d(a)=d$ for any $a\in\cala$, and    
define $\sigma^*_{\sf d}\in\distr(\calad)$ as $\sigma^*_{\sf d}(f_d)\eqdef \delta^*_\alpha(d)$ for any $d\in\cald$.
Let $\Pay(\delta,\alpha) = \postvf{\pi}{\VChoice{a}{\alpha}\, {\HChoice{d}{\delta}} {C_{d a}}}$.
Then $\Pay(\delta,\alpha)$ is convex in $\delta$ and linear in $\alpha$.
Hence we have:
\begin{align*}
\min_\delta\max_\alpha\postvf{\pi}{\VChoice{a}{\alpha}\, {\HChoice{d}{\delta}} {C_{d a}}} 
=&\,\, 
\max_\alpha\min_\delta\postvf{\pi}{\VChoice{a}{\alpha}\, {\HChoice{d}{\delta}} {C_{d a}}}
&\text{(by Theorem~\ref{theo:vonneumann})}
\\ =&\,\, 
\max_\alpha\postvf{\pi}{\VChoice{a}{\alpha}\, {\HChoice{d}{\delta^*_\alpha}} {C_{d a}}}
&\text{(by definition of $\delta^*_\alpha$)}
\\ =&\,\, 
\max_\alpha\postvf{\pi}{\VChoice{a}{\alpha}\, {\HChoice{f_d}{\sigma^*_{\sf d}}} {C_{f_d(a) a}}}
&\text{(by definition of $\sigma^*_{\sf d}$)}
\\ \geq&\,\, 
\max_\alpha  \min_{\msd} {\postvf{\pi}{\VChoice{a}{\alpha}\;\HChoice{\psd}{\msd}  C_{\psd(a) a}} }
& \text{(by minimization on $\msd$)}
\end{align*}
\end{proof}

\subsection{Visible choice vs. hidden choice}
We consider now the case of Games III and IV-V. In the running example the payoff for III is lower than for IV-V, but it is easy to find other cases in which the situation is reversed. For instance, if in the running example we set $C_{11}$ to be the same as $C_{01}$, the payoff for  III  will be ${1}$ (corresponding to the choice  $a=1$ for the attacker), and that for IV-V will be $\nicefrac{2}{3}$ (corresponding to the Nash equilibrium $p^*=q^*=\nicefrac{2}{3}$. So we conclude that Games III and IV-V are incomparable: there is no general ordering between them. 

The relation between Games I and IV comes from the fact that they are both simultaneous  games,  and 
the only difference is the way in which the payoff is defined. 
The same holds for the case of Games III and VI$_{\sf m}$, which are both attacker-first sequential games.   
The essence of the proof  is expressed by the following proposition.
\begin{restatable}[visible choice $\geqslant$ hidden choice]{Proposition}{order}
\label{prop:order:I-IV and III-VI}
For every $a\in \cala$ and every $\delta\in \distr\cald$ we have:\,\,
$
\postvf{\pi}{{\VChoice{d}{\delta}} {C_{d a}}}\,\, \geq\,\, 
\postvf{\pi}{{\HChoice{d}{\delta}} {C_{d a}}}.
$
\end{restatable}
\begin{proof}
\begin{align*}
\postvf{\pi}{{\VChoice{d}{\delta}} {C_{d a}}}
 \,\, =&\,\,  \sum_{d \in \cald} \delta(d) \,\postvf{\pi}{C_{d a}}
& \text{(by Theorem~\ref{theo:convex-V-q}(\ref{enumerate:vg:v:linear}))}
\\ \geq&\,\, \,
\postvf{\pi}{{\HChoice{d}{\delta}} {C_{d a}}}
& \text{(by Theorem~\ref{theo:convex-V-q}(\ref{enumerate:vg:h:convex}))}
\end{align*}
\end{proof}

From the above proposition we can derive immediately the following corollaries:

\begin{restatable}[Game I $\geqslant$ Game IV]{Corollary}{order}
\label{theo:order:I-IV}
\begin{align*}
\min_\delta\max_\alpha \postvf{\pi}{\VChoiceDouble{d}{\delta}{a}{\alpha} C_{da}}
\;\; \geq\;\;
\min_\delta\max_\alpha\postvf{\pi}{\VChoice{a}{\alpha}\, {\HChoice{d}{\delta}} {C_{d a}}}
\end{align*}
\end{restatable}

\begin{restatable}[Game III $\geqslant$ Game VI$_{\sf m}$]{Corollary}{order}
\label{theo:order:III-VI}
\begin{align*}
\max_\alpha  \min_{\msd}{\postvf{\pi}{\VChoiceDouble{\psd}{\msd}{a}{\alpha}  C_{\psd(a) a}} }
\;\; \geq\;\;
\max_\alpha  \min_{\msd} {\postvf{\pi}{\VChoice{a}{\alpha}\;\HChoice{\psd}{\msd}  C_{\psd(a) a}} }
\end{align*}
\end{restatable}

 Finally, we show that the vulnerability for the optimal solution in Game VI$_{\sf m}$  
 is always greater than or equal to that of Game VI$_{\sf b}$, which means 
 that for the defender it is always convenient to use behavioral strategies.  
 We can state actually a more general result: for any mixed strategy, there is always a behavioral strategy that gives the same payoff. 
 \begin{Proposition}\label{prop:mixed-behavioral}
 For any $\alpha\in \distr\cala$ and any $\msd \in\distr(\calad)$ there exists $\fsd :\cala\rightarrow \distr(\cald)$ such that:
\[
\postvf{\pi}{ \VChoice{a}{\alpha}\;\HChoice{d}{\fsd(a)}  C_{d a} } 
\,=\,
\postvf{\pi}{ \VChoice{a}{\alpha}\;\HChoice{\psd}{\msd}  C_{\psd(a) a} } 
\]
\end{Proposition}
\begin{proof}
For $\msd \in\distr(\calad)$, define $\fsd :\cala\rightarrow \distr(\cald)$ as follows:  for any $a\in\cala$ and $d\in\cald$, 
\[\displaystyle \fsd(a)(d) \eqdef \sum_{\psd(a)=d}\msd (\psd)\] and observe that for every $a\in \cala$ we have
$\displaystyle \HChoice{d}{\fsd(a)}  C_{d a}  = \HChoice{\psd}{\msd}  C_{\psd(a) a}$. 
\end{proof}
From this proposition we derive immediately the following corollary:
\begin{Corollary}\label{cor:mixed-behavioral}
\[
\max_\alpha  \min_{\msd} {\postvf{\pi}{\VChoice{a}{\alpha}\;\HChoice{\psd}{\msd}  C_{\psd(a) a}} }
\,\geq \,
\max_\alpha  \min_{\fsd} {\postvf{\pi}{\VChoice{a}{\alpha}\;\HChoice{d}{\fsd(a)}  C_{d a}} }
\]
\end{Corollary}

\begin{figure}[tb]
\centering
\includegraphics[width=0.45\columnwidth]{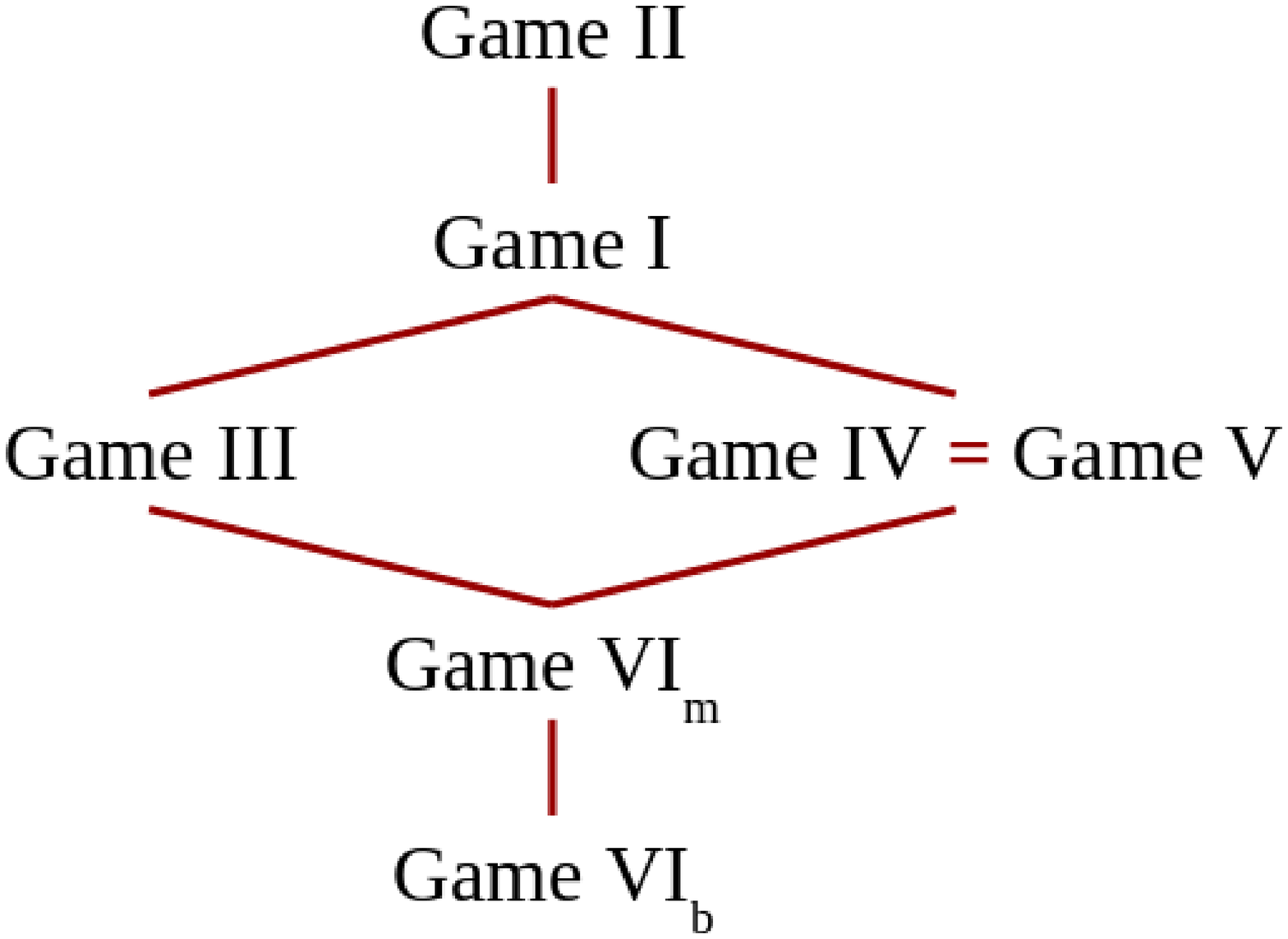}
\caption{Order of games w.r.t. the payoff in the Nash equilibrium. Games higher in the lattice have larger payoff.}
\label{fig:order-games1}
\end{figure}

The lattice in Figure~\ref{fig:order-games1} illustrates the results of this section about the relations between the various games.
These relations  can be used by the defender as guidelines to better protect the system, if he has some control over the rules of the game. Obviously, for the defender the games lower in the ordering are to be preferred to  compose protocols, since they yield a lower vulnerability for the result.

\section{Case study: a safer, faster password-checker}
\label{sec:password-example}
In this section we apply our game-theoretic, compositional
approach to show how a defender can mitigate an attacker's typical timing 
side-channel attack, while avoiding the usual burden imposed on the 
password-checker's efficiency by measures that make time consumption constant.

The following sections are organized as follows:
We first provide a formalization of the trade-off between efficiency and security
in password checkers using our framework of leakage games.
We then illustrate the approach in a simple instance of the program for $3$-bit passwords.
Finally, we provide general results for the $n$-bit case regarding the defender's optimal strategy
in equilibrium.

\subsection{Modeling the trade-off between efficiency and security as a game}

Consider the password-checker \texttt{PWD}\textsubscript{1..n} of 
Algorithm~\ref{alg:pwd-checker}, which performs a bitwise-check of
an $n$-bit low-input $a = a_{1} a_{2} \ldots a_{n}$ provided by the attacker
against an $n$-bit secret password $x = x_{1} x_{2} \ldots x_{n}$.
The bits are compared in increasing order ($1$, $2$, $\ldots$, $n$), 
with the low-input being rejected as soon as it mismatches the secret, 
and accepted otherwise.

\begin{center}
\begin{minipage}{0.5\linewidth}
\IncMargin{1em}
\begin{algorithm}[H]
\SetKwData{Left}{left}\SetKwData{This}{this}\SetKwData{Up}{up}
\SetKwFunction{Union}{Union}\SetKwFunction{FindCompress}{FindCompress}
\SetKwInOut{InputH}{High input}
\SetKwInOut{InputL}{Low input}
\SetKwInOut{Output}{Output}
\SetKw{Return}{return}
\InputH{\,$x \in \{0,1\}^{n}$}
\InputL{\,$a \in \{0,1\}^{n}$}
\Output{\,$y \in \{\true,\false\}$}
\BlankLine
{\it accept} $:= \true$\;
\For{$i := 1$ \KwTo $n$}{
\If(){$a_{i} \neq x_{i}$}{\label{lt}
{\it accept} $:= \false$\;
\textbf{break}\;
}
}
\Return{accept}\;
\caption{Password-checker \texttt{PWD}\textsubscript{1..n}.}
\label{alg:pwd-checker}
\end{algorithm}
\DecMargin{1em}
\end{minipage}
\end{center}

The attacker can choose low-inputs to try to gain information about the password. 
Obviously, in case \texttt{PWD}\textsubscript{1..n} accepts the low-input,
the attacker learns the password value is $a = x$.
Yet, even when the low-input is rejected,
there is some leakage of information:
from the duration of the execution the attacker can estimate how 
many iterations have been performed before the low-input was rejected, 
thus inferring a prefix of the secret password. 

To model this scenario, let $\calx = \{0,1\}^{n}$ be the set of all possible $n$-bit 
passwords, and 
$\caly = \{\false\}\times\{1,2,\ldots,n\} \cup \{\true\} \times \{n\} = \{(\false,1), (\false,2), (\false,3), \ldots, (\false, n), (\true,n)\}$ 
be the set of observables produced by the system.
Each observable is an ordered pair whose first element indicates whether or not 
the password was accepted ($\true$ or $\false$, respectively), and the second
element indicates the duration of the computation ($1$, $2$, $\ldots$, or $n$ iterations).

For instance, consider a scenario with $3$-bit passwords.
Let \texttt{PWD}\textsubscript{123} be a password checker that performs 
the bitwise comparison in increasing order ($1$, $2$, $3$).
Channel $C_{123,101}$ in Table~\ref{table:pwd-channels-123}
models \texttt{PWD}\textsubscript{123}'s behavior when the attacker provides 
low-input $a=101$.
Note that this channel represents the fact that \texttt{PWD}\textsubscript{123}
accepts the low-input when the secret is $x=101$ 
(the channel outputs $(\true,3)$ with probability $1$),
and otherwise rejects the low-input in a certain number of steps 
(e.g., the checker rejects the low-input in $2$ steps when the password is $x=110$, 
so in this case the channel outputs $(\false,2)$ with probability $1$).

\begin{table}[tb]
\centering
\vspace{-4mm}
\begin{subtable}[t]{0.45\linewidth}
\centering
\caption{Channel $C_{123,101}$ modeling the case in which the defender 
compares bits in order (1,2,3), and the attacker picks 
low-input $101$.}
$
\begin{array}{|c|c|c|c|c|}\hline
\multirow{2}{*}{$C_{123,101}$} & y =  & y =  & y =  & y =  \\
 & (\false,1) & (\false,2) & (\false,3) & (\true,3) \\ \hline
x = 000    & 1 & 0 & 0 & 0 \\
x = 001    & 1 & 0 & 0 & 0 \\
x = 010    & 1 & 0 & 0 & 0 \\
x = 011    & 1 & 0 & 0 & 0 \\
x = 100    & 0 & 0 & 1 & 0 \\
x = 101    & 0 & 0 & 0 & 1 \\
x = 110    & 0 & 1 & 0 & 0 \\
x = 111    & 0 & 1 & 0 & 0 \\ \hline
\end{array}
$
\label{table:pwd-channels-123}
\end{subtable}
\qquad
\begin{subtable}[t]{0.45\linewidth}
\centering
\caption{Channel $C_{\text{cons},101}$ modeling the case in which the defender 
runs a constant-time checker, and the attacker picks low-input $101$.}
$
\begin{array}{|c|c|c|}
\hline
\multirow{2}{*}{$C_{\text{cons},101}$} & y =  & y =  \\ 
& (\false,3) & (\true,3) \\ \hline
x = 000    & 1 & 0  \\
x = 001    & 1 & 0  \\
x = 010    & 1 & 0  \\
x = 011    & 1 & 0  \\
x = 100    & 1 & 0  \\
x = 101    & 0 & 1  \\
x = 110    & 1 & 0  \\
x = 111    & 1 & 0  \\ \hline
\end{array}
$
\label{table:pwd-channels-cons}
\end{subtable}
\caption{Channels $C_{da}$ modeling different versions of $3$-bit password
checker for defender's action $d$ and attacker's action $a$.}
\label{table:pwd-channels}
\end{table}

To quantify the password checker's leakage of information, we will adopt
Bayes vulnerability, so the prior Bayes vulnerability 
$\priorvf{\pi}$ will correspond to the probability that the attacker guesses correctly
the password in one try, whereas the posterior Bayes vulnerability 
$\postvf{\pi}{C}$ will correspond to the probability that the attacker guesses correctly
the password in one try, after he observes the output of the channel (i.e.,
after he has measured the time needed for the checker to accept or reject 
the low-input).
For instance, in the $3$-bit password scenario, if the prior distribution 
on all possible $3$-bit passwords is 
$\hat{\pi}=(0.0137, 0.0548, 0.2191, \allowbreak 0.4382, \allowbreak 0.0002, \allowbreak 0.0002, \allowbreak 0.0548, \allowbreak 0.2191)$, 
the corresponding prior Bayes vulnerability is 
$\priorvf{\hat{\pi}} = 0.4382$.
For prior $\hat{\pi}$ above, the posterior Bayes vulnerability of channel 
$C_{123,101}$ is $\postvf{\hat{\pi}}{C_{123,101}} = 0.6577$,
which represents an increase in Bayes vulnerability 
of about $50\%$).

A way to mitigate this timing side-channel is to make the checker's
execution time independent of the secret. 
This can be done by by eliminating the \texttt{break} command within 
the loop in \texttt{PWD}\textsubscript{1..n}, so no matter when the
matching among high and low input happens, the password checker will
always need $n$ iterations to complete.
%
For instance, in the context of our $3$-bit password example, we can let \texttt{PWD}\textsubscript{cons} be a constant-time 
$3$-bit password checker that applies this counter measure.
Channel $C_{\text{cons},101}$ from Table~\ref{table:pwd-channels-cons} 
models \texttt{PWD}\textsubscript{cons}'s behavior when the attacker's 
low-input is $a = 101$.
Note that this channel reveals only whether or not the low-input
matches the secret value, but does not allow the attacker to infer
a prefix of the password.
Indeed, this channel's posterior Bayes vulnerability is 
$\postvf{\hat{\pi}}{C_{123,101}} = 0.4384$, 
which brings the multiplicative Bayes leakage down to an increase of only 
about $0.05\%$.

However, the original program is substantially more efficient than the
modified one. Consider the general case of $n$-bit passwords and assume that
either the password, or the program's low input, are chosen uniformly at
random. Because of this assumption, each bit being checked in the original
program has probability $\nicefrac{1}{2}$ of being rejected. Hence, the
program will finish after $1$ iteration with probability $\nicefrac{1}{2}$,
after $2$ iterations with probability $\nicefrac{1}{4}$, and so on, up to the
$n$-th iteration. After that the program always finishes, so with the
remaining probability $2^{-n}$ the program finishes after $n$ iterations,
giving a total expected time
of
\[
	\sum_{k=1}^{n} k 2^{-k}  + n 2^{-n} = 2(1-2^{-n})
	~.
\]
The above derivation is based on the series $\sum_{k=1}^n k z^k =
z\frac{1-(n+1)z^n+nz^{n+1}}{(1-z)^2}$. Hence the expected running time of the
original program (under the uniform assumption) is constant: always bounded
by $2$, and converging to $2$ as $n$ grows. On the other hand, the running
time of the modified constant-time program is $n$ iterations, an $\nicefrac{n}{2}$-fold
increase.

Seeking some compromise between security and efficiency, assume 
the defender can employ different versions of the password-checker, 
each performing the bitwise  comparison among low-input $a$ and secret 
password $x$ in a different order.
More precisely, there is one version of the checker for every 
possible order in which the index $i$ 
ranges in the control of the loop in Algorithm~\ref{alg:pwd-checker}.

To determine a defender's best choice of which versions of the checker to run,
we model this problem as game.
The attacker's set of actions $\cala$
consists in all possible $2^{n}$ low-inputs to the checker, 
and the defender's set of actions
$\cald$ consists in all $n!$ orders in which the checker can perform the bitwise 
comparison.
There is, then, 
a channel $C_{ad}{:}\calx{\times}\caly{\rightarrow}\reals$ for
each possible combination of $d \in \cald$, $a \in \cala$.
In our framework, the payoff of a mixed strategy profile
$(\delta, \alpha)$ is given by
$
\Pay(\delta,\alpha) \allowbreak = \allowbreak \expectDouble{a\leftarrow\alpha}{}\hspace{-1ex} \; \postvf{\pi}{{\HChoice{d}{\delta}} {C_{d a}}}.
$
For each pure strategy profile $(d,a)$, the payoff of the game will be the posterior 
Bayes vulnerability
of the resulting channel $C_{da}$
(since, if we are measuring the information leakage, the prior vulnerability is the same
for every channel once the prior is fixed).

In the $3$-bit password scenario, the attacker's actions 
$\cala = \{000, 001, 010, 011, 100, 101, 110, 111\}$ are all
possible $3$-bit low-inputs, and the defender's 
$\cald = \{123, 132, 213, 231, 312, 321 \}$ are all 
possible versions of the password checker (each action represents the order in which the
$3$ bits are checked).
Table~\ref{tab:payoff} depicts the corresponding payoffs of all
$48$ possible resulting channels $C_{ad}$ with $d \in \cald$, $a \in \cala$,
when the prior is still $\hat{\pi}=(0.0137, 0.0548, 0.2191, \allowbreak 0.4382, \allowbreak 0.0002, \allowbreak 0.0002, \allowbreak 0.0548, \allowbreak 0.2191)$.
Note that the attacker's and defender's actions substantially affect
the effectiveness of the attack:
vulnerability ranges between 0.4934 and 0.9311 
(and so multiplicative leakage is in the range
between an increase of $12\%$ and one of $112\%$).
Using techniques from~\cite{Alvim:17:GameSec}, we can compute
the best (mixed) strategy for the defender in this game, which turns out 
to be
\begin{align*}
\delta^{*}  = (0.1667, 0.1667, 0.1667, 0.1667, 0.1667, 0.1667).
\end{align*}
\begin{table}[tb]
\centering
$
\begin{array}{c|c||c|c|c|c|c|c|c|c|}
\multicolumn{2}{c}{} & \multicolumn{8}{c}{\textbf{Attacker's action $a$}} \\ \cline{2-10}
	& U(d,a) & 000 & 001 & 010 & 011 & 100 & 101 & 110 & 111 \\ \cline{2-10} \cline{2-10}
\multirow{6}{*}{\rotatebox[origin=c]{90}{\stackanchor{\textbf{Defender's}}{\textbf{action $d$}}}}
	& 123 &  0.7257 & 0.7257 & 0.9311 & 0.9311 & 0.6577 & 0.6577 & 0.7122 & 0.7122 \\ \cline{2-10}
	& 132 &  0.8900 & 0.9311 & 0.8900 & 0.9311 & 0.7122 & 0.7122 & 0.7122 & 0.7122 \\ \cline{2-10}
	& 213 &  0.5068 & 0.5068 & 0.9311 & 0.9311 & 0.4934 & 0.4934 & 0.7668 & 0.7668 \\ \cline{2-10}
	& 231 &  0.5068 & 0.5068 & 0.7668 & 0.9311 & 0.5068 & 0.5068 & 0.7668 & 0.9311 \\ \cline{2-10}
	& 312 &  0.7257 & 0.9311 & 0.7257 & 0.9311 & 0.7122 & 0.8766 & 0.7122 & 0.8766 \\ \cline{2-10}
	& 321 &  0.6712 & 0.7122 & 0.7257 & 0.9311 & 0.6712 & 0.7122 & 0.7257 & 0.9311 \\ \cline{2-10}
\end{array}
$
\caption{Payoff for each pure strategy profile of $3$-bit password scenario.}
\label{tab:payoff}
\end{table}
This strategy is part of an equilibrium and guarantees that for any choice
of the attacker the posterior Bayes vulnerability is at most $0.6573$
(so the multiplicative leakage is bounded by $50\%$, an intermediate value
between the minimum of about $12\%$ and the maximum of about $112\%$).

The running time, on the other hand, of this new password-checker is the same
as that of the original one. Under the assumption that either the
password, or the low-input, are uniformly distributed, each check fails with
probability $\nicefrac{1}{2}$, giving a total expected time of
$2(1-2^{-n})$. Hence this technique substantially decreases the program's information leakage,
without affecting at all its expected running time.

\subsection{On optimal strategies for the defender}

Interestingly, in the $3$-bit password case study from the previous section,
the defender's optimal strategy consists in uniformly sampling among all
available versions of the checker. A uniform distribution seems to be an
adequate candidate for the defender, but is it always the best choice
for any prior and any number of bits?

We first answer this question in the case of a uniform prior $\pi$ for arbitrary
$n$-bit passwords, which already
turns out to be challenging. Under this prior, and exploiting a crucial symmetry of the password
checker (see the proof of Theorem~\ref{prop:uniform-prior-delta}),
we can show that all strategies for the adversary are in fact equivalent, namely
\[
	\Pay(\alpha,\delta) = \Pay(\alpha',\delta)
	\quad\text{for all } \alpha,\alpha',\delta
	~.
\]
For the defender, on the other hand, the situation is far from trivial:
although all \emph{pure} strategies $d$ are still equivalent,
$\Pay(\alpha,\delta)$ does in general depend on $\delta$.
By exploiting another symmetry of the password checker together with the symmetry
of $\vf$, we can show that a uniform strategy is indeed optimal for the
defender, as stated in the following result.

\begin{restatable}{Theorem}{uniformpriordelta}
\label{prop:uniform-prior-delta}
	Consider the password checker program of
	Algorithm~\ref{alg:pwd-checker} for $n$-bit passwords, where the attacker controls
	the low input to the checker, and the defender controls the
	order in which the bits are checked.
	If the
	prior $\pi$ on possible passwords is uniform, and the payoff is given
	by the posterior Bayes-vulnerability
	$
	\Pay(\delta,\alpha) \allowbreak = \allowbreak 
	\expectDouble{a\leftarrow\alpha}{}\hspace{-1ex} \; \postvf{\pi}{{\HChoice{d}{\delta}} {C_{d a}}}
	$,
	then the strategy $(\delta^*,\alpha)$ where $\delta^*$ is \emph{uniform}
	and $\alpha$ is \emph{arbitrary} is
	an equilibrium strategy.
\end{restatable}

Perhaps surprisingly, however, Theorem~\ref{prop:uniform-prior-delta} 
does not generalize to non-uniform priors (or to different vulnerability metrics).
More precisely, when the prior on passwords is not uniform, the defender
may benefit from assigning different probabilities to different versions of the
password checker. 
This subtlety arises from the fact that the defender's goal is not to maximize
the attacker's uncertainty about the selected password checker itself 
(i.e., the defender's action), but it is rather to maximize the attacker's 
uncertainty about the secret value.
The following examples illustrate this (perhaps counter-intuitive) fact.

Consider again a $3$-bit password scenario, similar to 
that of the previous section. 
Assume that the attacker knows only that the first bit of the password is 
surely $0$, so that the prior on secrets is
$$
\pi^{(A)} = (0.25, 0.25, 0.25, 0.25, 0, 0, 0, 0).
$$
The payoff table for this case is presented in Table~\ref{tab:3bit-counterexample-piA},
and a corresponding equilibrium defender's best strategy, computed using techniques
from \cite{Alvim:17:GameSec}, is
$$
\delta^{*(A)} = (0.25, 0.25, 0, 0.25, 0, 0.25).
$$
Note that this equilibrium means that the defender never has to check
the bits in the order (2,1,3) or in the order (3,1,2).
The game's payoff (i.e., posterior Bayes vulnerability) in this case is 
$0.5625$, which is smaller than the payoff of $0.5833$ that would ensue
in case the defender's strategy were uniform. 
This means that, from the point of view of the defender, uniformly
randomizing is not optimal.

Now assume that the attacker knows that some passwords are more likely 
than others, even if all are possible, as reflected in the prior
$$
\pi^{(B)} = (0.25, 0.20, 0.15, 0.10, 0.10, 0.10, 0.05, 0.05).
$$
The payoff table for this case is presented in Table~\ref{tab:3bit-counterexample-piB},
and a corresponding equilibrium defender's best strategy can be computed to be
$$
\delta^{*(B)} = (0.1974, 0.1974, 0.1316, 0.1316,0.1711, 0.1711).
$$
Note that this equilibrium means that every version of the checker
may be selected by the defender, but the probability distribution is 
not uniform.
The game's payoff (i.e., posterior Bayes vulnerability) in this case is 
$0.4553$, which is again smaller than the payoff of $0.4666$ that would ensue
in case the defender's strategy were uniform. 

\begin{table}[tb]
\centering
\begin{subtable}[t]{\linewidth}
\centering
\caption{Under prior $\pi^{(A)}$.}
$
\begin{array}{c|c||c|c|c|c|c|c|c|c|}
\multicolumn{2}{c}{} & \multicolumn{8}{c}{\textbf{Attacker's action $a$}} \\ \cline{2-10}
  & U(d,a) & 000 & 001 & 010 & 011 & 100 & 101 & 110 & 111 \\ \cline{2-10} \cline{2-10}
\multirow{6}{*}{\rotatebox[origin=c]{90}{\stackanchor{\textbf{Defender's}}{\textbf{action $d$}}}}
  & 123 & 0.75 & 0.75 & 0.75 & 0.75 & 0.25 & 0.25 & 0.25 & 0.25  \\ \cline{2-10}
  & 132 & 0.75 & 0.75 & 0.75 & 0.75 & 0.25 & 0.25 & 0.25 & 0.25  \\ \cline{2-10}
  & 213 & 0.75 & 0.75 & 0.75 & 0.75 & 0.50 & 0.50 & 0.50 & 0.50  \\ \cline{2-10}
  & 231 & 0.75 & 0.75 & 0.75 & 0.75 & 0.75 & 0.75 & 0.75 & 0.75  \\ \cline{2-10}
  & 312 & 0.75 & 0.75 & 0.75 & 0.75 & 0.50 & 0.50 & 0.50 & 0.50  \\ \cline{2-10}
  & 321 & 0.75 & 0.75 & 0.75 & 0.75 & 0.75 & 0.75 & 0.75 & 0.75  \\ \cline{2-10}
\end{array}
$
\label{tab:3bit-counterexample-piA}
\end{subtable}
\\[5mm]
\begin{subtable}[t]{\linewidth}
\centering
\caption{Under prior $\pi^{(B)}$.}
$
\begin{array}{c|c||c|c|c|c|c|c|c|c|}
\multicolumn{2}{c}{} & \multicolumn{8}{c}{\textbf{Attacker's action $a$}} \\ \cline{2-10}
  & U(d,a) & 000 & 001 & 010 & 011 & 100 & 101 & 110 & 111 \\ \cline{2-10} \cline{2-10}
\multirow{6}{*}{\rotatebox[origin=c]{90}{\stackanchor{\textbf{Defender's}}{\textbf{action $d$}}}}
  & 123 & 0.70 & 0.70 & 0.60 & 0.60 & 0.50 & 0.50 & 0.45 & 0.4 \\ \cline{2-10}
  & 132 & 0.70 & 0.65 & 0.70 & 0.65 & 0.50 & 0.50 & 0.50 & 0.5 \\ \cline{2-10}
  & 213 & 0.70 & 0.70 & 0.55 & 0.55 & 0.60 & 0.60 & 0.50 & 0.5 \\ \cline{2-10}
  & 231 & 0.70 & 0.70 & 0.55 & 0.55 & 0.70 & 0.70 & 0.55 & 0.5 \\ \cline{2-10}
  & 312 & 0.70 & 0.65 & 0.70 & 0.65 & 0.60 & 0.60 & 0.60 & 0.6 \\ \cline{2-10}
  & 321 & 0.70 & 0.65 & 0.65 & 0.60 & 0.70 & 0.65 & 0.65 & 0.6 \\ \cline{2-10}
\end{array}
$
\label{tab:3bit-counterexample-piB}
\end{subtable}
\caption{Payoff tables, under different priors, for each pure strategy profile of 
$3$-bit password scenario.}
\label{tab:3bit-counterexample}
\end{table}

\section{Related work}
\label{sec:related-work}
Many studies have applied game theory to analyses of security 
and privacy in networks~\cite{Basar:83:TIT,Grossklags:08:WWW,Alpcan:11:TMC},
cryptography~\cite{Katz:08:TCC},
anonymity~\cite{Acquisti:03:FC}, 
location privacy~\cite{Freudiger:09:CCS}, and
intrusion detection~\cite{Zhu:09:GAMENETS},
to cite a few.
See \cite{Manshaei:13:ACMCS} for a survey. 

In the context of quantitative information flow, 
most works consider only passive attackers.
Boreale and Pampaloni~\cite{Boreale:15:LMCS} consider adaptive
attackers, but not adaptive defenders, and show that 
in this case the attacker's optimal strategy can 
be always deterministic.
Mardziel et al.~\cite{Mardziel:14:SP} propose a model for both adaptive 
attackers and defenders, but in none of their extensive case-studies 
the attacker needs a probabilistic strategy to maximize leakage.
In this paper we characterize when randomization is necessary, 
for either attacker or defender, to achieve optimality in our general 
information leakage games.

Security games have been employed to model and analyze payoffs between 
interacting agents, especially between a defender and an attacker.
Korzhyk et al. \cite{Korzhyk:11:JAIR} theoretically analyze security 
games and study the relationships between Stackelberg and Nash Equilibria 
under various forms of imperfect information.
Khouzani and Malacaria \cite{Khouzani:16:CSF} study leakage properties when perfect
secrecy is not achievable due to constraints on the allowable size of the conflating
sets, and provide universally optimal strategies for a wide class of entropy measures,
and for $g$-entropies.
In particular, they prove that designing a channel with minimum
leakage is equivalent to computing Nash equilibria in a corresponding 
two-player zero-sum games of incomplete information for a range of 
entropy measures.
These works, contrarily to ours, do not consider games with hidden 
choice, in which optimal strategies differ from traditional game-theory.

Several security games have modeled leakage when the sensitive
information are the defender's choices themselves, rather than a system's 
high input.
For instance, Alon et al.~\cite{Alon:13:SIAMDM} propose 
zero-sum games in which a defender chooses probabilities of secrets and an attacker chooses and learns some of the defender's secrets.
Then they present how the leakage on the defender's secrets gives influences on the defender's optimal strategy.
More recently, Xu et al.~\cite{Xu:15:IJCAI} show zero-sum games in which the attacker 
obtains partial knowledge on the security resources that the defender protects, and provide 
the defender's optimal strategy under the attacker's such knowledge.
Contrarily to these studies, in this paper we assume that a secret value is drawn from some prior distribution and is not the defender's strategy itself.

Security games have also been used to provide optimal trade-offs between two conflicting desirable properties. 
Khouzani et al.~\cite{Khouzani:15:CSF} study the clash
between security and  usability in the password selection, and present a game-theoretic 
framework for determining an optimal trade-off. They analyze guessing attacks and derive the optimal policies 
for secret picking as Nash/Stackelberg Equilibria.
Yang et al.~\cite{Yang:12:POST} propose a game-theoretic framework to analyze user behavior in anonymity networks. They consider the cost of anonymity in terms of the loss of utility. 
They also consider incentives and their impact on users' cooperation.
Shokri et al.~\cite{Shokri:17:TOPS} present a game-theoretic model for a designer to find the optimal privacy mechanism by taking the adversary's knowledge into account.
More specifically, they show a Stackelberg Bayesian game in which a user first chooses a location obfuscation mechanism to maximize his location privacy and then an adversary tries to estimate the user's location to minimize its error.
In contrast, our work presents a more general framework that is not limited to a particular domain, and focuses on protocol composition as a method to limit the leakage.

Regarding channel operators, sequential and parallel composition of channels
have been studied (e.g.,~\cite{Kawamoto:17:LMCS}), but we are unaware of any
explicit definition and investigation of hidden and visible choice operators.
Although Kawamoto et al.~\cite{KawamotoBL16} implicitly use the hidden choice to model 
a probabilistic system as the weighted sum of systems, they do not derive the set 
of algebraic properties we do for this operator, and for its interaction with the 
visible choice operator.

\section{Conclusion and future work}
\label{sec:conclusion}
In this paper we used protocol composition to model the introduction of noise 
performed by the defender to prevent leakage of sensitive information.
More precisely, we formalized visible and hidden probabilistic choices of different protocols.
We then formalized the interplay between defender and attacker in a game-theoretic framework adapted to the specific issues of QIF, where the payoff is information leakage. 
We considered various kinds of leakage games, depending on whether players act simultaneously or sequentially, and whether the choices of the defender are visible or not to the attacker.
We established a hierarchy of these games, and provided methods for finding the optimal strategies (at the points of equilibrium) in the various cases.
We also proved that in a sequential game with hidden choice, the behavioral strategies are more advantageous for the defender than the mixed strategies.
This contrast with the standard game theory, where the two types of strategies are equivalent.

As future research, we would like to extend leakage games to the 
case of repeated observations, i.e., when the attacker can 
observe  the outcomes of the system in successive runs,  under the 
assumption that both attacker and defender may change the 
channel in each run. 
We would also like to extend our framework to non zero-sum games, in which
the costs of attack and defense are not equivalent, and to analyze differentially-private mechanisms.

\vspace{6pt} 



\acknowledgments{
The authors are thankful to Arman Khouzani and Pedro O. S. Vaz de Melo 
for valuable discussions.
This work was supported by JSPS and Inria under the project LOGIS of the Japan-France AYAME Program,
by the PEPS 2018 project MAGIC, 
and by the project Epistemic Interactive Concurrency (EPIC) from the STIC 
AmSud Program.
M\'{a}rio S. Alvim was supported by CNPq, CAPES, and FAPEMIG.
Yusuke Kawamoto was supported by JSPS KAKENHI Grant Number JP17K12667.
}

\authorcontributions{
All authors contributed to the technical results, and are listed in alphabetical order.
}

\conflictsofinterest{The authors declare no conflict of interest.
} 

\abbreviations{The following abbreviations are used in this manuscript:\\

\noindent 
\begin{tabular}{@{}ll}
QIF & Quantitative information flow
\end{tabular}
}

\appendixtitles{yes} 
\appendixsections{multiple} 
\appendix
\section{Proofs of Technical Results}
\label{sec:proofs}

In this section we provide the proofs of the  technical results which are not in the main body of the paper.

\subsection{Preliminaries for proofs}

We start by providing some necessary background for
the subsequent technical proofs.

Definition~\ref{def:equivalence-channels} states that
two compatible channels (i.e., with the same input space)
are equivalent if yield the same value of vulnerability 
for every prior and every vulnerability function.
The result below, from the literature, provides 
necessary and sufficient conditions for two channels being 
equivalent.
The result employs the extension of channels with an all-zero 
column as follows.
For any channel $C$ of type $\calx \times \caly \rightarrow \reals$, 
its \emph{zero-column extension} $C^{0}$ is the channel of type 
$\calx \times \left(\caly \cup \{y_{0}\}\right) \rightarrow \reals$,
with $y_{0} \notin \caly$, s.t. $C^{0}(x,y) = C(x,y)$ for all 
$x\in\calx$, $y\in\caly$, and $C^{0}(x,y_{0}) = 0$ for all $x \in \calx$.

\begin{restatable}[Characterization of channel equivalence~\cite{Alvim:12:CSF,McIver:14:POST}]{Lemma}{resequivalentchannelsconditions}
\label{lemma:equivalent-channels-conditions}
Two channels $C_{1}$ and $C_{2}$ are equivalent iff
every column of $C_{1}$ is a convex combination of 
columns of $C_{2}^{0}$, and every column of $C_{2}$ is a 
convex combination of columns of $C_{1}^{0}$.
\end{restatable}

Note that the result above implies that for being equivalent,
any two channels must be compatible.

\subsection{Proofs of Section~\ref{sec:operators}}

\restypehiddenchoice*

\begin{proof}
Since hidden choice is defined as a summation
of matrices, the type of $\HChoice{i}{\mu} C_{i}$ is 
the same as the type of every $C_{i}$ in the family.

To see that $\HChoice{i}{\mu} C_{i}$ is a channel
(i.e., all of its entries are non-negative, and all
of its rows sum up to 1),
first note that, since each $C_{i}$ in the family is a 
channel matrix, $C_{i}(x,y)$ lies in the interval 
$[0,1]$ for all $x \in \calx$, $y\in \caly$.
Since $\mu$ is a set of convex coefficients, from 
the definition of hidden choice
it follows that
also $\sum_{i \in \cali} \mu(i)C_{i}(x,y)$ must lie 
in the interval $[0,1]$ for every $x,y$.

Second, note that for all $x \in \calx$:
\begin{align*}
\sum_{y\in \caly} \left(\HChoice{i}{\mu} C_{i}\right)(x,y)
=&\, \sum_{y\in \caly} \sum_{i \in \cali} \mu(i)C_{i}(x,y) & \text{(def. of hidden choice)} \\
=&\, \sum_{i \in \cali} \mu(i) \sum_{y\in \caly} C_{i}(x,y) \\
=&\, \sum_{i \in \cali} \mu(i) \cdot 1 & \text{($C_{i}{:}\calx{\times}\caly{\rightarrow}\reals$ are channels)} \\
=&\, 1 & \text{($\mu$ is a prob. dist.)}
\end{align*}
\end{proof}

\restypevisiblechoice*

\begin{proof}
Visible choice applied to a family $\{C_{i}\}$ of channels 
scales each matrix $C_{i}$ by a factor $\mu(i)$, which
preserves the type $\calx \times \caly_{i} \rightarrow \reals$ 
of each matrix, and then concatenates all the matrices so produced,
yielding a result of type
$\calx \times \left( \bigsqcup_{i \in \cali} \caly_{i} \right) \rightarrow \reals$.

To see that $\VChoice{i}{\mu} C_{i}$ is a channel
(i.e., that all of its entries are non-negative, and that
all rows sum-up to $1$), note that each element of the 
visible choice on the family $\{C_{i}\}$ can be denoted by 
$\left(\VChoice{i}{\mu} C_{i}\right)(x,(y,j))$,
where $x \in \calx$, $j \in \cali$, and $y \in \caly_{j}$.
Then, note that for all $x \in \calx$, $j \in \cali$, and $y \in \caly_{j}$:
\begin{align}
\label{eq:propvchoicetype0}
\left(\VChoice{i}{\mu} C_{i}\right)(x,(y,j)) \nonumber 
=&\, \left(\bigconc_{i \in \cali} \mu(i)C_{i}\right)(x,(y,j)) & \text{(def. of visible choice)} \nonumber \\
=&\, \left(\mu(j)C_{j}\right)(x,y) & \text{(def. of concatenation)} \nonumber \\
=&\, \mu(j)C_{j}(x,y) & \text{(def. of scalar mult.)}
\end{align}
which is a non-negative value, since, both $\mu(j)$ and 
$C_{j}(x,y)$ are non-negative.

Finally, note that for all $x \in \calx$:
\begin{align*}
\sum_{\substack{j \in \cali \\ y \in \caly_{j}}} \left(\VChoice{i}{\mu} C_{i}\right)(x,(y,j))
=&\, \sum_{\substack{j \in \cali \\ y \in \caly_{j}}} \mu(j)C_{j}(x,y) &
\text{(by Eq.~\eqref{eq:propvchoicetype0})} \\ 
=&\, \sum_{j \in \cali} \mu(j) \sum_{y \in \caly_{j}} C_{j}(x,y) &
\text{} \\ 
=&\, \sum_{j \in \cali} \mu(j) \cdot 1 & \text{($C_{j}{:}\calx{\times}\caly_{j}{\rightarrow}\reals$ are channels)} \\
=&\, 1 & \text{($\mu$ is a prob. dist.)}
\end{align*}
\end{proof}

\residempotency*

\begin{proof}

\begin{enumerate}
\renewcommand{\labelenumi}{\alph{enumi})}

\item Idempotency of hidden choice:
\begin{align*}
 \HChoice{i}{\mu}C_{i} 
=&\, \sum_{i} \mu(i)C_{i} & \text{(def. of hidden choice)} \\
=&\, \sum_{i} \mu(i)C & \text{(since every $C_{i}=C$)} \\
=&\, C \sum_{i} \mu(i) & \text{} \\
=&\, C & \text{(since $\mu$ is a prob. dist.)}
\end{align*}

\item Idempotency of visible choice:
\begin{align*}
 \VChoice{i}{\mu}C_{i} 
=&\, \bigconc_{i} \mu(i)C_{i} & \text{(def. of visible choice)} \\
=&\, \bigconc_{i} \mu(i)C & \text{(since every $C_{i}=C$)} \\
\equiv&\, C & \text{(by Lemma~\ref{lemma:equivalent-channels-conditions})}
\end{align*}

In the above derivation, we can apply 
Lemma~\ref{lemma:equivalent-channels-conditions}
because every column of the channel on each side of the equivalence 
can be written as a convex combination of the zero-column extension
of the channel on the other side.
\end{enumerate}
\end{proof}

\resreorganizationoperators*

\begin{proof}

\begin{enumerate}
\renewcommand{\labelenumi}{\alph{enumi})}

\item 
\begin{align*}
\HChoice{i}{\mu} \HChoice{j}{\eta} C_{i j} 
=&\, \HChoice{i}{\mu} \left( \sum_{j} \eta(j) C_{i j} \right) & \text{(def. of hidden choice)} \\ 
=&\, \sum_{i} \mu(i) \left( \sum_{j} \eta(j) C_{i j} \right) & \text{(def. of hidden choice)} \\ 
=&\, \sum_{i,j} \eta(i)\eta(j) C_{i j} & \text{(reorganizing the sums)} \\ 
=&\, \HChoiceDouble{i}{\mu}{j}{\eta} C_{i j} & \text{(def. of hidden choice)}
\end{align*}

\item 
\begin{align*}
\VChoice{i}{\mu} \VChoice{j}{\eta} C_{i j} 
=&\, \VChoice{i}{\mu} \left( \bigconc_{j} C_{i j} \right) & \text{(def. of visible choice)} \\ 
=&\, \bigconc_{i} \mu(i) \left( \bigconc_{j} \eta(j) C_{i j} \right) & \text{(def. of visible choice)} \\
\equiv& \bigconc_{i j} \mu(i) \eta(j) C_{i j} & \text{(by Lemma~\ref{lemma:equivalent-channels-conditions})} \\ 
=&\,\VChoiceDouble{i}{\mu}{j}{\eta} C_{i j} & \text{(def. of visible choice)}
\end{align*}
In the above derivation, we can apply 
Lemma~\ref{lemma:equivalent-channels-conditions}
because every column of the channel on each side of the equivalence 
can be written as a convex combination of the zero-column extension
of the channel on the other side.

\item 
\begin{align*}
\HChoice{i}{\mu} \VChoice{j}{\eta} C_{i j} 
=&\, \HChoice{i}{\mu} \left( \bigconc_{j} C_{i j} \right) & \text{(def. of visible choice)} \\ 
=&\, \sum_{i} \mu(i) \left( \bigconc_{j} \eta(j) C_{i j} \right) & \text{(def. of hidden choice)} \\ 
\equiv&\,\bigconc_{j} \eta(j) \left( \sum_{i} \eta(i) C_{i j} \right) & \text{(by Lemma~\ref{lemma:equivalent-channels-conditions})} \\ 
=&\, \VChoice{j}{\eta} \HChoice{i}{\mu} C_{i j} & \text{(def. of operators)}
\end{align*}
In the above derivation, we can apply 
Lemma~\ref{lemma:equivalent-channels-conditions}
because every column of the channel on each side of the equivalence 
can be written as a convex combination of the zero-column extension
of the channel on the other side.
\end{enumerate}

\end{proof}



%

\subsection{Proofs of Section~\ref{sec:password-example}}

\uniformpriordelta*

\begin{proof}
	We show that $(\delta^*,\alpha)$ where $\delta^*$ is uniform and $\alpha$ is
	arbitrary, is a saddle point of $\Pay(\delta,\alpha)$.
	The proof will rely on two crucial \emph{symmetries} of the password checker, in combination
	with the use of a uniform prior.
	Note that, under uniform $\pi$, the Bayes-vulnerability of a channel $C$
	is proportional to the sum of the column maxima of $C$, namely
	$\postvf{\pi}{C} = \frac{1}{|\calx|}\sum_y\max_x C(x,y)$.
	Given a permutation $\sigma$ of $\calx$, define $C^\sigma$ as $C$ with its
	rows permuted, namely $C^\sigma(x,y) = C(\sigma(x),y)$.
	Note that $C^\sigma$ can be written as $M_\sigma C$
	where $M_\sigma$ is a permutation matrix.
	Permuting the rows
	does not affect the column maxima, hence $\postvf{\pi}{C} = \postvf{\pi}{C^\sigma}$.

	For the attacker things are simple, since we can show that under a uniform
	$\pi$ all attacker strategies are equivalent, that is,
	\[
		\Pay(\alpha,\delta) = \Pay(\alpha',\delta)
		\quad\text{for all } \alpha,\alpha',\delta
		~.
	\]
	To show this, we use the first important symmetry of the password checker, which
	is due to the fact that the output of the
	algorithm only depends on whether $x_i \neq a_i$ is true or false for each bit.
	Hence, any modification on $a$ can be matched with a modification on $x$ that preserves the output,
	namely for all $a,a'$ there exists a permutation $\sigma$ of passwords
	such that $C_{da} = C_{da'}^\sigma$ for all $d$.
	Denote $C_{\delta a} = \HChoice{d}{\delta} C_{d a}$,
	we have that
	\[
		C_{\delta a}
		= \HChoice{d}{\delta} C_{d a}
		= \HChoice{d}{\delta} M_\sigma C_{d a'}
		= M_\sigma \HChoice{d}{\delta} C_{d a'}
		= C_{\delta a'}^\sigma
	\]
	From $\postvf{\pi}{C_{\delta a}} = \postvf{\pi}{C_{\delta a'}^\sigma}$ we have that
	$\Pay(a,\delta) = \Pay(a',\delta)$ for all
	pure strategies $a,a'$, 
	which implies $\Pay(\alpha,\delta) = \Pay(\alpha',\delta)$
	since $\Pay(\alpha,\delta)$ is linear on $\alpha$.
	So in the remaining of the proof we assume that the attacker plays the fixed pure strategy
	$a_0$, having the zero bitstring $0..0$ as the low input.

	For the defender, on the other hand, the situation is far from trivial:
	although all \emph{pure} strategies $d$ are still equivalent,
	$\Pay(\alpha,\delta)$ is \emph{only convex} on $\delta$ and as a consequence the
	payoff highly depends on $\delta$.
	Our goal is to show that
	\[
	\Pay(a_0,\delta)
	= \frac{1}{|\calx|}\sum_y\max_x \sum_d \delta(d) C_{da_0}(x,y)
	\]
	is minimized on the uniform $\delta^*$. To do so, we show something stronger, namely
	that \emph{each addend} of the $y$-sum is \emph{simultaneously} minimized on the uniform
	$\delta^*$. Fix an arbitrary $y$, and let
	\[
	f(\delta)
	= \max_x \delta \cdot \psi_x
	\]
	where $\cdot$ is the dot product and $\psi_x\in\reals^{|\cald|}$ is the vector defined by
	$\psi_x(d) = C_{da_0}(x,y)$.
	Note that, since $C_{da_0}$ is deterministic, all elements of $\psi_x$ are either $0$
	or $1$. Intuitively, $\psi_x(d)= 1$ if $x$ produces the fixed output $y$ when the
	bit checking order is $d$.

	We need to show that $f(\delta)$ is minimized on $\delta^*$.
	However, $f$ seen as a function on the whole $\reals^{|\cald|}$ has no global
	minimum. In our case, though, $\delta$ is a probability distribution
	taking values in $\distr\cald\subset \reals^{|\cald|}$. That is, only $|\cald|-1$ elements
	of $\delta$ are free, so we can reduce the dimension by one as follows: fix some $d_0\in\cald$ and let
	$\tilde\delta,\tilde\psi_x\in\reals^{|\cald|-1}$ be the same as 
	$\delta,\psi_x$ with the element corresponding to $d_0$ removed.
	We have that $\delta(d_0) = 1-\sum_{d\neq d_0}\delta(d) = 1 - \tilde\delta\cdot\mathbf{1}$,
	where $\mathbf{1}$ is the ``ones'' vector, hence we can rewrite $f(\delta)$ as a function
	of $\tilde\delta$:
	\begin{align*}
		f(\tilde\delta)
		&= \max_x \tilde\delta \cdot \tilde\psi_x + \delta(d_0)\psi_x(d_0)
			\\
		&= \max_x \tilde\delta \cdot \tilde\psi_x
		+ (1 - \tilde\delta\cdot\mathbf{1})\psi_x(d_0) 
			\\
		&= \max_x \tilde\delta \cdot (\tilde\psi_x - \psi_x(d_0)\mathbf{1}) + \psi_x(d_0)
	\end{align*}
	So it is sufficient to show that $f(\tilde\delta)$ is minimized on $\tilde\delta^*$.

	In the following, we use the fundamental concept of
	\emph{subgradients} from convex analysis, which generalize gradients for
	non-differentiable functions. A vector $v$ is a subgradient of
	a possibly non-differentiable convex function
	$g:S\to\reals$
	at $x_0\in S$ iff
	\[
		g(x) - g(x_0) \ge v \cdot (x - x_0)
		\qquad \text{for all } x\in	S
		~.
	\]
	It is well-known that $g$ has
	a global minimum on $x_0$ iff the $\mathbf{0}$ vector belongs to the set of
	subgradients of $g$ at $x_0$ (this generalizes the fact that differentiable convex
	functions have zero gradient on their global minimum). Recall also that $g(x) = x\cdot v + c$
	has a single subgradient $v$, while for $g(x)=\max_i g_i(x)$, any subgradient
	of $g_i$ for \emph{any branch $i$} giving the max, is also a subgradient of $g$.
	Finally, the set of subgradients is convex, so any convex combination of them
	is also a subgradient.

	Hence, the subgradients of $f(\tilde\delta)$ are $\tilde\psi_x -
	\psi_x(d_0)\mathbf{1}$, for any $x$ giving the maximum, as well as their
	convex combinations. Our goal is to show that on $\tilde\delta^*$ these
	include the zero vector.

	We finally arrive to the second important symmetry of the password checker.
	Let $\rho$ be a permutation of the set $\{1,\ldots,n\}$ of password bits.
	Such a permutation could be seen both as a permutation on $\calx$
	(i.e. $\rho(x) = x_{\rho(1)}.. x_{\rho(n)}$; note that $\rho(x)$ has the same
	number of $0$s and $1$s as $x$), as well as a permutation on
	$\cald$ (a bit checking order $d$ is \emph{itself} a permutation of bits,
	so we can set $\rho(d) = \rho\circ d$).
	Since all low bits of $a_0$ are the same,
	applying $\rho$ to both $x$ and $d$ does not change the outcome of the
	algorithm, that is $C_{da_0} = C_{\rho(d)a_0}^\rho$.

	Intuitively, if $d$ is selected uniformly, then the attacker cannot
	distinguish $x$ from $\rho(x)$ since they
	produce the same output with the same probability.
	More concretely, 
	let $x\sim x'$ denote that $x' = \rho(x)$ for some bit permutation $\rho$.
	Due to the aforementioned symmetry we have that 
	$\psi_x$ and $\psi_{x'}$ have the same number of $1$s which means that
	$\delta^*\cdot\psi_x = \delta^*\cdot\psi_{x'}$.
	Now, let $x^* \in \mathrm{argmax}_x \delta^*\cdot\psi_x$.
	We have that \emph{all} $x\sim x^*$ also give the $\max$, hence all vectors
	\begin{equation}\label{eq-password-1}
	\tilde\psi_x -
		\psi_x(d_0)\mathbf{1}
		\qquad x\sim x^*
	\end{equation}
	are subgradients of $f$ on $\tilde\delta^*$.
	Finally, for any $d,d'$ there is a bit permutation $\rho$ such that
	$d' = \rho(d)$. Since $\psi_x(d) = \psi_{\rho(x)}(\rho(d))$ we have
	that
	$
		(\sum_{x\sim x^*} \psi_x)(d) = k
	$
	(for some integer $k$) independently from $d$.
	Hence, letting $c = |\{x\ |\ x\sim x^*\}|$ and averaging all subgradients of
	\eqref{eq-password-1} we get that
	\[
		\frac{1}{c}
		\sum_{x\sim x^*} \big(\tilde\psi_x -
			\psi_x(d_0)\mathbf{1}\big)
		= \frac{1}{c}(k\mathbf{1} - k\mathbf{1})
		= \mathbf{0}
	\]
	is also a subgradient, which concludes the proof.
\end{proof}
\section{Properties of binary versions of channel operators}
\label{sec:operators-properties-binary}
In this section we derive some relevant properties of 
the binary versions of the hidden and visible choice 
operators.
We start with results regarding each operator
individually.

\begin{restatable}[Properties of the binary hidden choice]{prop}{respropertiesbinaryhiddenchoice}
\label{prop:properties-binary-hidden-choice}
For any channels $C_{1}$ and $C_{2}$ of the same type, 
and any values $0 \leq p,q \leq 1$, 
the binary hidden choice operator satisfies the following 
properties:
\begin{enumerate}
\renewcommand{\labelenumi}{\alph{enumi})}

\item \emph{Idempotency}:
$C_{1} \hchoice{p} C_{1} = C_{1}$
\item \emph{Commutativity}: 
$C_{1} \hchoice{p} C_{2} = C_{2} \hchoice{\overline{p}} C_{1}$
\item \emph{Associativity}: 
$$C_{1} \hchoice{p} (C_{2} \hchoice{q} C_{3}) = ( \nicefrac{1}{q} \cdot C_{1} \hchoice{p} C_{2}) \hchoice{q} \overline{p} \cdot C_{3}$$ 
if $q \neq 0$.
\item \emph{Absorption}: 
$$(C_{1} \hchoice{p} C_{2}) \hchoice{q} (C_{1} \hchoice{r} C_{2})
= C_{1} \hchoice{(pq + \overline{q}r)} C_{2}.$$
\end{enumerate}
\end{restatable}

\begin{proof}
We will prove each property separately.

\begin{enumerate}
\renewcommand{\labelenumi}{\alph{enumi})}

\item \emph{Idempotency}: 
\begin{align*}
C_{1} \hchoice{p} C_{1}
=&\, p \cdot C_{1} + \overline{p} \cdot C_{1} & \text{(def. of hidden choice)} \\
=&\, (p + \overline{p}) \cdot C_{1} & \text{} \\
=&\, C_{1} & \text{($p+\overline{p}=1$)}
\end{align*}

\item \emph{Commutativity}:
\begin{align*}
C_{1} \hchoice{p} C_{2}
=&\, p \cdot C_{1} + \overline{p} \cdot C_{2} & \text{(def. of hidden choice)} \\
=&\, \overline{p} \cdot C_{2} + p \cdot C_{1} & \text{}\\
=&\, C_{2} \hchoice{\overline{p}} C_{1} & \text{(def. of hidden choice)}
\end{align*}

\item \emph{Associativity}:
\begin{align*}
 &\, C_{1} \hchoice{p} ( C_{2} \hchoice{q} C_{3} ) \\
=&\, p \cdot C_{1} + \overline{p} ( q \cdot C_{2} + \overline{q} \cdot C_{3} ) & \text{(def. of hidden choice)} \\
=&\,p \cdot C_{1} + \overline{p}q \cdot C_{2} + \overline{p}\overline{q} \cdot C_{3} & \text{} \\
=&\,q(p (\nicefrac{1}{q} \cdot C_{1}) + \overline{p} \cdot C_{2}) + \overline{q}(\overline{p} \cdot C_{3}) & \text{} \\
=&\,(\nicefrac{1}{q} \cdot C_{1} \hchoice{p} C_{2}) \hchoice{q} \overline{p} \cdot C_{3} & \text{(def. of hidden choice)} \\
\end{align*}

\item \emph{Absorption}: First note that:
\begin{align*}
 &\,(C_{1} \hchoice{p} C_{2}) \hchoice{q} (C_{1} \hchoice{r} C_{2}) \\
=&\,q (p C_{1} + \overline{p} C_{2}) + \overline{q} (r C_{1} + \overline{r} C_{2}) & \text{(def. of hidden choice)} \\
=&\,pq C_{1} + \overline{p}q C_{2} + \overline{q}r C_{1} + \overline{q}\overline{r} C_{2} \\
=&\,(pq + \overline{q}r) C_{1} + (\overline{p}q + \overline{q}\overline{r}) C_{2} \\
=&\,C_{1} \hchoice{(pq + \overline{q}r)} C_{2} & \text{(*)}
\end{align*}
To complete the proof, note that in step (*) above,
$pq + \overline{q}r$ and $\overline{p}q + \overline{q}\overline{r}$ form
a valid binary probability distribution (they are both 
non-negative, and they add up to 1), then apply the
definition of hidden choice.

\end{enumerate}
\end{proof}

\begin{restatable}[Properties of binary visible choice]{prop}{respropertiesbinaryvisiblechoice}
\label{prop:properties-binary-visible-choice}
For any compatible channels $C_{1}$ and $C_{2}$, 
and any values $0 \leq p,q \leq 1$, 
the visible choice operator satisfies the following 
properties:
\begin{enumerate}
\renewcommand{\labelenumi}{\alph{enumi})}

\item \emph{Idempotency}: 
$C_{1} \vchoice{p} C_{1} \equiv C_{1}$
\item \emph{Commutativity}: 
$C_{1} \vchoice{p} C_{2} \equiv C_{2} \vchoice{\overline{p}} C_{1}$
\item \emph{Associativity}: 
$C_{1} \vchoice{p} (C_{2} \vchoice{q} C_{3}) \equiv ( \nicefrac{1}{q} \cdot C_{1} \vchoice{p} C_{2}) 
\vchoice{q} \overline{p} \cdot C_{3}$
if $q \neq 0$.
\end{enumerate}
\end{restatable}

\begin{proof}

We will prove each property separately.

\begin{enumerate}
\renewcommand{\labelenumi}{\alph{enumi})}

\item \emph{Idempotency}: $C_{1} \vchoice{p} C_{1} \equiv C_{1}$, by 
immediate application of Lemma~\ref{lemma:equivalent-channels-conditions},
since every column of the channel on each side of the equivalence 
can be written as a convex combination of the zero-column extension
of the channel on the other side.

\item \emph{Commutativity}:
\begin{align*}
C_{1} \vchoice{p} C_{2} 
=&\, p \cdot C_{1} \conc \overline{p} \cdot C_{2} & \text{(def. of visible choice)} \\
\equiv&\ \overline{p} \cdot C_{2} \conc p \cdot C_{1} & \text{(by Lemma~\ref{lemma:equivalent-channels-conditions})} \\
=&\ C_{2} \vchoice{\overline{p}} C_{1} & \text{(def. of visible choice).} \\
\end{align*}
In the above derivation, we can apply 
Lemma~\ref{lemma:equivalent-channels-conditions}
because every column of the channel on each side of the equivalence 
can be written as a convex combination of the zero-column extension
of the channel on the other side.

\item \emph{Associativity}:
\begin{align*}
C_{1} \vchoice{p} ( C_{2} \vchoice{q} C_{3} ) 
=&\, p \cdot C_{1} \conc \overline{p} ( q \cdot C_{2} \conc \overline{q} \cdot C_{3} ) & \text{(def. of visible choice)} \\
\equiv&\, p \cdot C_{1} \conc \left( \overline{p}q \cdot C_{2} \conc \overline{p}\overline{q} \cdot C_{3} \right) & \text{(by Lemma~\ref{lemma:equivalent-channels-conditions})} \\
=&\, q(p (\nicefrac{1}{q} \cdot C_{1}) \conc \overline{p} \cdot C_{2}) \conc \overline{q}(\overline{p} \cdot C_{3}) & \text{} \\
=&\,(\nicefrac{1}{q} \cdot C_{1} \vchoice{p} C_{2}) \vchoice{q} \overline{p} \cdot C_{3} & \text{(def. of visible choice)}
\end{align*}
In the above derivation, we can apply 
Lemma~\ref{lemma:equivalent-channels-conditions}
because every column of the channel on each side of the equivalence 
can be written as a convex combination of the zero-column extension
of the channel on the other side.

\end{enumerate}
\end{proof}

Now we turn our attention to the interaction between
hidden and visible choice operators.

A first result is that hidden choice does not distribute 
over visible choice.
To see why, note that $C_{1} \hchoice{p} (C_{2} \vchoice{q} C_{3})$
and $(C_{1} \hchoice{p} C_{2}) \vchoice{q} (C_{1} \hchoice{p} C_{3})$
cannot be both defined: if the former is defined, then $C_{1}$ must have
the same type as $C_{2} \vchoice{q} C_{3}$, whereas if the
latter is defined, $C_{1}$ must have the same type as $C_{2}$, but
$C_{2} \vchoice{q} C_{3}$ and $C_{2}$ do not have the same type
(they have different output sets).

However, as the next result shows, visible choice distributes over 
hidden choice.

\begin{restatable}[Distributivity of $\vchoice{p}$ over $\hchoice{q}$]{prop}{resdistributivity}
\label{prop:distributivity}
Let $C_{1}$, $C_{2}$ and $C_{3}$ be compatible channels,
and $C_{2}$ and $C_{3}$ have the same type.
Then, for any values $0 \leq p,q \leq 1$:
\begin{align*}
C_{1} \vchoice{p} (C_{2} \hchoice{q} C_{3}) \equiv (C_{1} \vchoice{p} C_{2}) \hchoice{q} (C_{1} \vchoice{p} C_{3}).
\end{align*}
\end{restatable}

\begin{proof}
\begin{align*}
 &\, C_{1} \vchoice{p} (C_{2} \hchoice{q} C_{3}) \\
=&\, (C_{1} \hchoice{q} C_{1}) \vchoice{p} (C_{2} \hchoice{q} C_{3}) & \text{(idempotency of visible choice)} \\
=&\, p (q \cdot C_{1} + \overline{q} \cdot C_{1}) \conc \overline{p}(q \cdot C_{2} + \overline{q} \cdot C_{3}) & \text{(def. of operators)} \\
=&\, (pq \cdot C_{1} + p\overline{q} \cdot C_{1}) \conc (\overline{p}q \cdot C_{2} + \overline{p}\overline{q} \cdot C_{3}) & \text{} \\
\equiv&\,  (pq \cdot C_{1} \conc \overline{p}q \cdot C_{2}) + (p\overline{q} \cdot C_{1} \conc \overline{p}\overline{q} \cdot C_{3}) & \text{(by Lemma~\ref{lemma:equivalent-channels-conditions})} \\
=&\, q(p \cdot C_{1} \conc \overline{p} \cdot C_{2}) + \overline{q}(p \cdot C_{1} \conc \overline{p} \cdot C_{3}) & \text{} \\
=&\, q(C_{1} \vchoice{p} C_{2}) + \overline{q}(C_{1} \vchoice{p} C_{3})  & \text{(def. of visible choice)} \\
=&\, (C_{1} \vchoice{p} C_{2}) \hchoice{q} (C_{1} \vchoice{p} C_{3})  & \text{(def. of hidden choice)}
\end{align*}
In the above derivation, we can apply 
Lemma~\ref{lemma:equivalent-channels-conditions}
because every column of the channel on each side of the equivalence 
can be written as a convex combination of the zero-column extension
of the channel on the other side.
\end{proof}

\externalbibliography{yes}
\bibliography{short,new}

\end{document}